%% file: soslearning.tex
\documentclass[12pt,letterpaper]{article}

\usepackage{comment}

\def\showauthornotes{0}
\def\showtableofcontents{1}
\def\showkeys{0}
\def\showdraftbox{0}
\def\showcolorlinks{1}
\def\usemicrotype{1}
\def\showfixme{0}

\input{macros}

\usepackage{bbm}

\usepackage{tikz}
\usetikzlibrary{matrix}

\let\pref=\prettyref

\let\super\varsuper

\renewcommand{\Ind}{\mathds 1}

\DeclareMathOperator*{\pE}{\tilde {\mathbb E}}

\newcommand{\sleq}{\preceq}
\newcommand{\sgeq}{\succeq}

\newcommand{\cor}{\mathrm{Cor}}
\title{Dictionary Learning and Tensor Decomposition via the Sum-of-Squares Method}

\author{%
Boaz Barak\thanks{Microsoft Research.}
\and Jonathan A.\ Kelner\thanks{Department of Mathematics, Massachusetts Institute of Technology.}
\and David Steurer\thanks{Department of Computer Science, Cornell University.}
}

\newcommand{\nips}[2]{#2}
\newcommand{\stoceq}[1]{$$#1$$}
\begin{document}

\maketitle
\draftbox
\thispagestyle{empty}

\input{abstract}
\clearpage

\ifnum\showtableofcontents=1
{
\tableofcontents
\thispagestyle{empty}
 }
\fi

\clearpage

\input{introduction}

\input{overview}

\input{prelim}

\input{dictionary}

\input{sampling}

\input{tensor-decomposition}

\input{polytime}

\input{conclusion}

\addreferencesection
\bibliographystyle{amsalpha}
\bibliography{bib/soslearning,bib/mr,bib/dblp,bib/orig-dblp,bib/scholar}
\appendix
\include{monomial-inequality}
\end{document}

%% file: macros.tex
\usepackage{etex}

\usepackage[l2tabu, orthodox]{nag}

\usepackage{xspace,enumerate}

\usepackage[dvipsnames]{xcolor}

\usepackage[american]{babel}

\usepackage{mathtools}

\usepackage{amsthm}

\newtheorem{theorem}{Theorem}[section]
\newtheorem*{theorem*}{Theorem}

\newtheorem*{proposition*}{Proposition}
\newtheorem{lemma}[theorem]{Lemma}
\newtheorem*{lemma*}{Lemma}

\newtheorem*{conjecture*}{Conjecture}

\newtheorem*{fact*}{Fact}

\newtheorem*{hypothesis*}{Hypothesis}

\theoremstyle{definition}
\newtheorem{definition}[theorem]{Definition}

\theoremstyle{remark}
\newtheorem{claim}[theorem]{Claim}
\newtheorem*{claim*}{Claim}
\newtheorem{remark}[theorem]{Remark}
\newtheorem*{remark*}{Remark}

\newtheorem*{observation*}{Observation}

\usepackage[
letterpaper,
top=1in,
bottom=1in,
left=1in,
right=1in]{geometry}

\usepackage[varg]{pxfonts} %
\usepackage{textcomp} %
\usepackage[scr=rsfso]{mathalfa}%
\usepackage{bm} %
\ifnum\showkeys=1
\usepackage[color]{showkeys}
\fi

\ifnum\showcolorlinks=1
\usepackage[
pagebackref,
colorlinks=true,
urlcolor=blue,
linkcolor=blue,
citecolor=OliveGreen,
]{hyperref}
\fi

\ifnum\showcolorlinks=0
\usepackage[
pagebackref,
colorlinks=false,
pdfborder={0 0 0}
]{hyperref}
\fi

\usepackage{prettyref}

\newcommand{\savehyperref}[2]{\texorpdfstring{\hyperref[#1]{#2}}{#2}}

\newrefformat{eq}{\savehyperref{#1}{\textup{(\ref*{#1})}}}
\newrefformat{lem}{\savehyperref{#1}{Lemma~\ref*{#1}}}
\newrefformat{def}{\savehyperref{#1}{Definition~\ref*{#1}}}
\newrefformat{thm}{\savehyperref{#1}{Theorem~\ref*{#1}}}
\newrefformat{cor}{\savehyperref{#1}{Corollary~\ref*{#1}}}
\newrefformat{cha}{\savehyperref{#1}{Chapter~\ref*{#1}}}
\newrefformat{sec}{\savehyperref{#1}{Section~\ref*{#1}}}
\newrefformat{app}{\savehyperref{#1}{Appendix~\ref*{#1}}}
\newrefformat{tab}{\savehyperref{#1}{Table~\ref*{#1}}}
\newrefformat{fig}{\savehyperref{#1}{Figure~\ref*{#1}}}
\newrefformat{hyp}{\savehyperref{#1}{Hypothesis~\ref*{#1}}}
\newrefformat{alg}{\savehyperref{#1}{Algorithm~\ref*{#1}}}
\newrefformat{rem}{\savehyperref{#1}{Remark~\ref*{#1}}}
\newrefformat{item}{\savehyperref{#1}{Item~\ref*{#1}}}
\newrefformat{step}{\savehyperref{#1}{step~\ref*{#1}}}
\newrefformat{conj}{\savehyperref{#1}{Conjecture~\ref*{#1}}}
\newrefformat{fact}{\savehyperref{#1}{Fact~\ref*{#1}}}
\newrefformat{prop}{\savehyperref{#1}{Proposition~\ref*{#1}}}
\newrefformat{prob}{\savehyperref{#1}{Problem~\ref*{#1}}}
\newrefformat{claim}{\savehyperref{#1}{Claim~\ref*{#1}}}
\newrefformat{relax}{\savehyperref{#1}{Relaxation~\ref*{#1}}}
\newrefformat{red}{\savehyperref{#1}{Reduction~\ref*{#1}}}
\newrefformat{part}{\savehyperref{#1}{Part~\ref*{#1}}}

\newcommand{\Sref}[1]{\hyperref[#1]{\S\ref*{#1}}}

\usepackage{nicefrac}

\let\ffrac=\flatfrac

\ifnum\usemicrotype=1
\usepackage{microtype}
\fi

\ifnum\showauthornotes=1
\newcommand{\Authornote}[2]{{\sffamily\small\color{red}{[#1: #2]}}}
\newcommand{\Authornotecolored}[3]{{\sffamily\small\color{#1}{[#2: #3]}}}
\newcommand{\Authorcomment}[2]{{\sffamily\small\color{gray}{[#1: #2]}}}
\newcommand{\Authorstartcomment}[1]{\sffamily\small\color{gray}[#1: }

\newcommand{\Authorfnote}[2]{\footnote{\color{red}{#1: #2}}}
\newcommand{\Authorfixme}[1]{\Authornote{#1}{\textbf{??}}}
\newcommand{\Authormarginmark}[1]{\marginpar{\textcolor{red}{\fbox{\Large #1:!}}}}
\else
\newcommand{\Authornote}[2]{}
\newcommand{\Authornotecolored}[3]{}
\newcommand{\Authorcomment}[2]{}
\newcommand{\Authorstartcomment}[1]{}

\newcommand{\Authorfnote}[2]{}
\newcommand{\Authorfixme}[1]{}
\newcommand{\Authormarginmark}[1]{}
\fi

\newcommand{\Dnote}{\Authornote{D}}

\newcommand{\Dcomment}{\Authorcomment{D}}

\newcommand{\Dfnote}{\Authorfnote{D}}

\newcommand{\Bnote}{\Authornote{B}}

\ifnum\showfixme=0

\fi

\usepackage{boxedminipage}

\newcommand{\paren}[1]{(#1)}
\newcommand{\Paren}[1]{\left(#1\right)}

\newcommand{\card}[1]{\lvert#1\rvert}

\newcommand{\Set}[1]{\left\{#1\right\}}

\newcommand{\norm}[1]{\lVert#1\rVert}

\newcommand{\iprod}[1]{\langle#1\rangle}

\newcommand{\Esymb}{\mathbb{E}}
\newcommand{\Psymb}{\mathbb{P}}

\DeclareMathOperator*{\E}{\Esymb}

\DeclareMathOperator*{\ProbOp}{\Psymb}

\renewcommand{\Pr}{\ProbOp}

\usepackage{dsfont}
\usepackage{mathrsfs}

\newcommand{\textparen}[1]{\text{(#1)}}

\ifx\because\undefined
\newcommand{\because}[1]{\textparen{because #1}}
\else
\renewcommand{\because}[1]{\textparen{because #1}}
\fi

\newcommand{\sge}{\succeq}

\newcommand{\super}[2]{#1^{\paren{#2}}}

\newcommand{\seteq}{\mathrel{\mathop:}=}

\newcommand{\mper}{\,.}
\newcommand{\mcom}{\,,}

\newcommand\bdot\bullet

\ifx\mathds\undefined %
\newcommand{\Ind}{\mathbb I}
\else
\newcommand{\Ind}{\mathds 1}
\fi

\DeclareMathOperator{\poly}{poly}

\DeclareMathOperator{\polylog}{polylog}

\newcommand{\Hoelder}{H\"{o}lder\xspace}

\newcommand{\N}{\mathbb N}
\newcommand{\R}{\mathbb R}

\newcommand{\cA}{\mathcal A}

\newcommand{\cD}{\mathcal D}
\newcommand{\cE}{\mathcal E}

\newcommand{\cL}{\mathcal L}

\newcommand{\cR}{\mathcal R}

\renewcommand{\leq}{\leqslant}
\renewcommand{\le}{\leqslant}
\renewcommand{\geq}{\geqslant}
\renewcommand{\ge}{\geqslant}

\ifnum\showdraftbox=1
\newcommand{\draftbox}{\begin{center}
  \fbox{%
    \begin{minipage}{2in}%
      \begin{center}%
          \Large\textsc{Working Draft}\\%
        Please do not distribute%
      \end{center}%
    \end{minipage}%
  }%
\end{center}
\vspace{0.2cm}}
\else
\newcommand{\draftbox}{}
\fi

\let\epsilon=\varepsilon

\numberwithin{equation}{section}

\newcommand\MYcurrentlabel{xxx}

\newcommand{\MYstore}[2]{%
  \global\expandafter \def \csname MYMEMORY #1 \endcsname{#2}%
}

\newcommand{\MYload}[1]{%
  \csname MYMEMORY #1 \endcsname%
}

\newcommand{\MYnewlabel}[1]{%
  \renewcommand\MYcurrentlabel{#1}%
  \MYoldlabel{#1}%
}

\newcommand{\MYdummylabel}[1]{}

\newcommand{\torestate}[1]{%
  \let\MYoldlabel\label%
  \let\label\MYnewlabel%
  #1%
  \MYstore{\MYcurrentlabel}{#1}%
  \let\label\MYoldlabel%
}

\newcommand{\restatetheorem}[1]{%
  \let\MYoldlabel\label
  \let\label\MYdummylabel
  \begin{theorem*}[Restatement of \prettyref{#1}]
    \MYload{#1}
  \end{theorem*}
  \let\label\MYoldlabel
}

\newcommand{\restatelemma}[1]{%
  \let\MYoldlabel\label
  \let\label\MYdummylabel
  \begin{lemma*}[Restatement of \prettyref{#1}]
    \MYload{#1}
  \end{lemma*}
  \let\label\MYoldlabel
}

\newcommand{\restateprop}[1]{%
  \let\MYoldlabel\label
  \let\label\MYdummylabel
  \begin{proposition*}[Restatement of \prettyref{#1}]
    \MYload{#1}
  \end{proposition*}
  \let\label\MYoldlabel
}

\newcommand{\restatefact}[1]{%
  \let\MYoldlabel\label
  \let\label\MYdummylabel
  \begin{fact*}[Restatement of \prettyref{#1}]
    \MYload{#1}
  \end{fact*}
  \let\label\MYoldlabel
}

\newcommand{\restate}[1]{%
  \let\MYoldlabel\label
  \let\label\MYdummylabel
  \MYload{#1}
  \let\label\MYoldlabel
}

\newcommand{\addreferencesection}{
  \phantomsection
  \addcontentsline{toc}{section}{References}
}

\newcommand{\e}{\epsilon}

\let\origparagraph\paragraph
\renewcommand{\paragraph}[1]{\origparagraph{#1.}}

\allowdisplaybreaks

\usepackage{paralist}

%% file: abstract.tex
\begin{abstract}

\Dnote{}

We give a new approach to the dictionary learning (also known as ``sparse coding'') problem of recovering an unknown $n\times m$ matrix $A$ (for $m \geq n$) from examples of the form
\[
y = Ax + e,
\]
where $x$ is a random vector in $\mathbb R^m$ with at most $\tau m$ nonzero coordinates, and $e$ is a random noise vector in $\mathbb R^n$ with bounded magnitude.
For the case $m=O(n)$, our algorithm recovers every column of $A$ within arbitrarily good constant accuracy in time $m^{O(\log m/\log(\tau^{-1}))}$, in particular achieving
polynomial time if $\tau = m^{-\delta}$ for any $\delta>0$, and time $m^{O(\log m)}$ if $\tau$ is (a sufficiently small) constant.
Prior algorithms with comparable assumptions on the distribution required the vector $x$ to be much sparser---at most $\sqrt{n}$ nonzero coordinates---and there were intrinsic barriers preventing these algorithms from applying for denser $x$.

We achieve this by designing an algorithm for \emph{noisy tensor decomposition} that can recover, under quite general conditions, an approximate rank-one decomposition of a tensor $T$, given access to a tensor $T'$ that is $\tau$-close to $T$ in the spectral
norm (when considered as a matrix). To our knowledge, this is the first algorithm for tensor decomposition that works in the constant spectral-norm noise regime, where there is no guarantee
that the local optima of $T$ and $T'$ have similar structures.

Our algorithm is based on a novel approach to using and analyzing the \emph{Sum of Squares} semidefinite programming hierarchy (Parrilo 2000, Lasserre 2001), and it can be viewed as an indication of the utility of this very general and powerful tool for unsupervised learning problems.

\end{abstract}


%% file: introduction.tex
\section{Introduction}
\label{sec:introduction}

The \emph{dictionary learning} (also known as ``sparse coding'') problem is to recover an unknown $n\times m$ matrix $A$ (known as a ``dictionary'') from examples of the form
\begin{equation}
y = Ax + e  \;, \label{eq:sparse-coding:noise}
\end{equation}
where $x$ is sampled from a distribution over  \emph{sparse} vectors in $\R^m$ (i.e., with much fewer than $m$ nonzero coordinates),
and $e$ is sampled from a distribution over noise vectors in $\R^n$ of some bounded magnitude.

This problem has found applications in multiple areas, including computational neuroscience~\cite{OlshausenF97,OlshausenF96,OlshausenF96b}, machine learning~\cite{ArgyriouEP06,RanzatoBL2007}, and computer vision and image processing~\cite{EladA2006,MairalLBHP2008,YangWHY2008}.
The appeal of this problem is that, intuitively, data should be sparse in the ``right'' representation
 (where every coordinate corresponds to a meaningful feature),
and finding this representation can be a useful first step for further processing, just as representing sound or image data in the Fourier or wavelet bases is often a very useful preprocessing step in signal or image processing.
See \cite{SpielmanWW12,AgarwalA0NT13,AroraGM13,AroraBGM14} and the references therein for further discussion of the history and motivation of this problem.

This is a nonlinear problem, as both $A$ and $x$ are unknown, and dictionary learning is a computationally challenging task  even in the noiseless case.
When $A$ is known, recovering $x$ from $y$ constitutes the \emph{sparse recovery / compressed sensing} problem,
which has efficient algorithms~\cite{Donoho2006,CandesRT2006}.
Hence, a common heuristic for dictionary learning is to use alternating minimization, using sparse recovery to obtain a guess for $x$ based on a guess of $A$, and vice versa.

Recently there have been several works giving dictionary learning algorithms with rigorous guarantees on their performance~\cite{SpielmanWW12,AgarwalA0NT13,AgarwalAN13,AroraGM13,AroraBGM14}.
These works differ in various aspects, but they all share a common feature: they give no guarantee of recovery unless the distribution $\{ x \}$ is over \emph{extremely
sparse} vectors, namely having less than $O(\sqrt{n})$ (as opposed to merely $o(n)$) nonzero coordinates.
(There have been other works dealing with the less sparse case, but only at the expense of making strong assumptions on $x$ and/or $A$;
see Section~\ref{sec:relworks} for more discussion of related works.)

In this work we give a different algorithm that can be proven to approximately recover the matrix $A$ even when $x$ is much denser (up to $\tau n$ coordinates for some small constant $\tau>0$ in some settings).
The algorithm works (in the sense of approximate recovery) even with noise, in  the so-called overcomplete case (where $m>n$), and without making any incoherence assumptions on the dictionary.

Our algorithm is based on the \emph{Sum of Squares} (SOS) semidefinite programming hierarchy~\cite{Shor87,Nesterov00,Parrilo00,Lasserre01}.
The SOS algorithm is a very natural method for solving non-convex optimization problems that has found applications in a variety of scientific fields, including control theory~\cite{henrion2005positive}, quantum
information theory~\cite{doherty2002distinguishing}, game theory~\cite{Parrilo2006polynomial}, formal verification~\cite{Harrison2007}, and more.
Nevertheless, to our knowledge this work provides the first rigorous bounds on the SOS algorithm's running time for a natural unsupervised learning problem.

\subsection{Problem definition and conditions on coefficient distribution}
\label{sec:intro:nice}
\label{sec:intro-nice}

In this section we formally define the dictionary learning problem and state our result.
We define a \emph{$\sigma$-dictionary} to be an $m\times n$ matrix $A=(a^1|\cdots|a^m)$ such that
$\norm{a^i}=1$ for all $i$, and $A^\top A \sleq \sigma I$ (where $I$ is the identity matrix).
The parameter $\sigma$ is an analytical proxy for the overcompleteness $m/n$ of the dictionary $A$.
In particular, if the columns of $A$ are in isotropic position (i.e., $A^\top A$ is proportional to the identity),
then the top eigenvalue of $A^\top A$ is its trace divided by $n$, which equals $(1/n)\sum_i \norm{a^i}^2  = m/n$
because all of the $a^i$'s have unit norm.\nips{}{\footnote{%
While we do not use it in this paper, we note that in the dictionary learning problem it is always possible to learn a linear ``whitening transformation'' $B$ from the samples that would place the columns in isotropic position, at the cost of potentially changing the norms of the vectors.
 (There also exists a linear transformation that keeps the vectors normalized~\cite{Barthe98,Forster01}, but we do not know how to learn it from the samples.)}}
In this work we are mostly interested in the case $m=O(n)$, which corresponds to $\sigma = O(1)$.

\paragraph{Nice distributions}
Let  $\{ x \}$ be some distribution over the coefficients in (\ref{eq:sparse-coding:noise}).
We will posit some conditions on low-order moments of $\{ x \}$ to allow recovery.
Let $d$ be some even constant that we will use as a parameter (think of $d=O(1)$).
Consider a $0/1$ vector $x\in \R^m$ with $\tau m$ nonzero coordinates.
Then $\tfrac{1}{m}\sum_{k\in [m]} x_i^d = \tau$ and $(\tfrac{1}{m}\sum x_i^{d/2})^2 = \tau^2$.
In other words, if we select three ``typical''  coordinates $i,j,k$, then
\begin{equation}
x_i^{d/2} x_j^{d/2} \leq \tau x_k^{d/2} \label{eq:single-vec-sparse}.
\end{equation}

Equation (\ref{eq:single-vec-sparse}) will motivate us in defining an analytical proxy for the condition that the distribution $\{ x \}$ over coefficients is $\tau$-sparse.\footnote{%
By using an analytical proxy as opposed to requiring strict sparsity, we are only enlarging the set of distributions under consideration.
However, we will make some additional conditions below, and in particular requiring low order non-square moments to vanish, that although seemingly mild compared to prior works, do restrict the family of distributions.}

Specifically, in the dictionary learning case,
since we are interested in learning \emph{all} column vectors, we want every coordinate $i$ to be typical (for example, if the coefficient $x_i$ is always $0$ or always $1$,
we will not be able to learn the corresponding column vector).
Moreover, a necessary condition for recovery is that every \emph{pair} of coordinates is somewhat typical in the sense that the events that $x_i$ and $x_j$ are nonzero are not perfectly correlated.
Indeed, suppose for simplicity that when $x_i$ is nonzero, it is distributed like an independent standard Gaussian.
Then if those two events were perfectly correlated, recovery would be impossible since the distribution over examples would be identical
if we replaced  $\{a_i,a_j \}$ with any pair  of vectors $\{ \Pi a^i, \Pi a^j  \}$ where $\Pi$ is a rotation in the plane spanned by $\{ a^i,a^j\}$.

Given these considerations, if we normalize our distribution so that $\E x_i^d = 1$ for all $i$,
then it makes sense to assume:\footnote{Our results generalize to the case where $\E x_i^d \in [c,C]$ for some constants $C>c>0$.}
\begin{equation}
\E x_i^{d/2} x_j^{d/2} \leq \tau \;, \label{eq:nice}
\end{equation}
for all $i\neq j$ and for some $\tau \ll 1$.

We can assume without loss of generality that the marginal distribution $\{x_i\}$ is symmetric around zero
(namely $\Pr[x_i = a] = \Pr[ x_i = -a]$ for all $a$), since given two samples $y=Ax+e$ and $y'=Ax'+e'$ we can treat them
as a single sample $y-y'=A(x-x')+e-e'$, and the distribution $x-x'$, which is only slightly less sparse (and slightly more noisy), has this property.
In particular this means we can assume $\E x_i^{2k+1}=0$ for every integer $k$ and $i\in [m]$.
We will strengthen this condition to assume that
\begin{equation}
 \E x^{\alpha} = 0 \label{eq:odd-vanish}
\end{equation} for every non-square monomial $x^\alpha$ of degree at most $d$.
(Here, $\alpha \in \{0,1,\ldots \}^m$ is a \emph{multiindex} and $x^{\alpha}$ denotes the monomial $\prod_i x_i^{\alpha_i}$.
The degree of $x^\alpha$ is $|\alpha|\seteq \sum_i \alpha_i$;
we say that $x^{\alpha}$ is non-square if $x^\alpha$ is not the square of another monomial, i.e.,, if $\alpha$ has an odd coordinate.)

We say that a distribution $\{ x \}$ is \emph{$(d,\tau)$-nice} if it satisfies (\ref{eq:nice}) and (\ref{eq:odd-vanish}).\footnote{We
will also assume that $\E x_i^{2d} \leq n^{c}$ for some constant $c$. This is a very mild assumption, and in some qualitative sense is necessary to avoid pathological cases such as a distribution that outputs the all zero vector
with probability $1-n^{-\omega(1)}$.}
One example for a $(d,\tau)$-nice distribution is the \emph{Bernoulli-Gaussian} distribution, where $x_i = y_iz_i$ with the $y_i$'s being independent $0/1$ random variables satisfying
$\Pr[ y_i = 1] = \tau$ and the  $z_i$'s being independent normally distributed random variables (normalized to satisfy $\E z_i^d = 1/\tau$).
Indeed, in this case, since (using Cauchy-Schwarz) $\E z_i^{d/2} z_j^{d/2} \leq \sqrt{\E z_i^d \E z_j^d} = 1/\tau$,
\[
\E x_i^{d/2} x_j^{d/2} = (\E y_i y_j) (\E z_i^{d/2} z_j^{d/2}) \leq \tau^2 (1/\tau) = \tau \;.
\]

In fact, we can replace here the normal distribution with any distribution satisfying $\E z_i^d =1$,  and also allow some dependence between the variables (in particular encapsulating the models considered by~\cite{AroraGM13}).
As our discussion above and this example demonstrates, the parameter $\tau$ serves as a proxy to the sparsity of $\{ x \}$, where a $(d,\tau)$-nice distribution $\{ x \}$ roughly corresponds to a distribution having at most $\tau n$ coordinates with significant mass.
(For technical reasons, our formal definition of nice distributions, Definition~\ref{def:nice}, is somewhat different but is qualitatively equivalent to the above\nips{}{, see Remark~\ref{rem:niceness-definition}}.)

\begin{remark} Another way to justify this notion of nice distributions is that, as our analysis shows, it is a natural way to ensure that if $a$ is a column of the dictionary then the random variable
$\iprod{a,y}$ for a random sample $y$ from (\ref{eq:sparse-coding:noise}) will be ``spiky'' in the sense that it will have a large $d$-norm compared to its $2$-norm. Thus it is a fairly clean way to enable recovery, especially in the setting (such as ours) where we don't assume orthogonality or even incoherence between the dictionary vectors.
\end{remark}

\paragraph{Modeling noise}
Given a noisy dictionary learning example of the form $y = Ax +e$, one can also view it (assuming we are in the non-degenerate case of $A$ having full rank) as $y= A(x+e')$ for some $e'$ (whose magnitude
is controlled by the norm of $e$ and the condition number of $A$). If $e'$ has sufficiently small magnitude, and is composed of i.i.d random variables (and even under more general conditions), the distribution $\{ x + e' \}$ will be nice as well. Therefore, we will not explicitly model the noise in the following, but rather treat it as part of the distribution $\{ x \}$ which our definition allows to be only ``approximately sparse''.

\subsection{Our results}

Given samples of the form $\{ y = Ax \}$ for a $(d,\tau)$-nice $\{x \}$, with $d$ a sufficiently large constant (corresponding to
having $\tau n$ nonzero entries), we can approximately recover the dictionary $A$ in polynomial
time as long as $\tau \leq n^{-\delta}$ for some $\delta>0$, and in quasipolynomial time as long as $\tau$ is a sufficiently small constant.
Prior polynomial-time algorithms required the distribution to range over vectors with less than $\sqrt{n}$ nonzero entries
 (and it was not known how to improve upon this even using quasipolynomial time).

We define the \emph{correlation} of a pair of vectors $u$, $a$, to be $\cor(u,a) = \iprod{u,a}^2/(\norm{u}\norm{a})^2$.
We say that two sets $S,T$ of vectors are \emph{$\e$-close} if for every $s\in S$ there is $t\in T$ such that $\cor(s,t) \geq 1-\e$, and for every $t\in T$ there is $s\in S$ such that $\cor(s,t) \geq 1-\e$.\footnote{%
This notion corresponds to the sets $\{ s/\norm{s} : s\in S\} \cup \{ -s/\norm{s} : s\in S\}$
and  $\{ t/\norm{t} : t\in T\}\cup \{ -t/\norm{t} : t\in T\}$  being close in Hausdorff distance, which makes sense in our setting, since we can only hope to recover the dictionary columns up to permutation and scaling.}

\begin{theorem}[Dictionary learning] \label{ithm:dict-learn}
For every $\e>0,\sigma \geq 1, \delta>0$ there exists $d$ and a
polynomial-time algorithm $\cR$ such that
for every $\sigma$-dictionary $A=(a^1|\cdots|a^m)$ and $(d,\tau=n^{-\delta})$-nice $\{ x\}$,
given $n^{O(1)}$ samples from from $\{ y = Ax \}$,
$\cR$ outputs with probability at least $0.9$ a set that is $\e$-close to $\{ a^1,\ldots, a^m \}$.
\end{theorem}

The hidden constants in the $O(\cdot)$ notation may depend on $\e,\sigma, \delta$.
The algorithm can recover the dictionary vectors even in the relatively dense case when  $\tau$ is (a sufficiently small) constant, at the expense of a quasipolynomial
(i.e., $n^{O(\log n)}$) running time.
See Theorems~\ref{thm:dict-learn} and~\ref{thm:dict-learn-poly} for a precise statement of the dependencies between the constants.

\begin{remark} Our algorithm aims to recover the vectors up to $\e$-accuracy, with a running time as in a PTAS that depends (polynomially) on $\e$ in the exponent.  Prior algorithms achieving exact recovery needed to assume much stronger conditions, such as incoherence of dictionary columns.
Because we have not made incoherence assumptions and have only assumed the signals obey an analytic notion of sparsity, exact recovery is not possible, and there are limitations on how precisely one can recover the dictionary vectors (even information theoretically).

We believe that it is important to understand the extent to which dictionary recovery can be performed with only weak assumptions on the model, particularly given that real-world signals are often only approximately sparse and have somewhat complicated distributions of errors.
When stronger conditions are present that make better error guarantees possible, our algorithm can provide an initial solution for local search methods (or other recovery algorithms) to boost the approximate solution to a more precise one.
We believe that understanding the precise tradeoffs between the model assumptions, achievable precision, and running time  is an interesting question for further research.

We also note that approximate recovery is directly useful in some applications (e.g., for learning applications one might only need to know if a feature, which is modeled by the magnitude of $\iprod{a,y}$ for a dictionary column $a$, is ``on'' or ``off'' for a particular sample $y$. By an averaging argument, for a typical sample $y$ and feature $a$, the
events that $\iprod{a,y}$ us large and  $\iprod{\Tilde{a},y}$ is large would have about $1-\e$ correlation, where $\Tilde{a}$ is the approximation we produce for $a$.
\end{remark}

Our main tool is a new algorithm for the \emph{noisy tensor decomposition problem}, which is of interest in its own right.
This is the problem of recovering the set $\{ a^1,\ldots, a^m \}$ of vectors given access to a noisy version of the polynomial
$\sum_{i=1}^m \iprod{a^i,u}^d = \norm{A^\top u}_d^d$ in $\R[u]$, where $A=(a^1|\cdots|a^m)$ is an $n\times m$ matrix.\footnote{%
For a vector $v\in\R^m$ and $p\geq 1$, we define $\norm{v}_p = (\sum_i |v_i|^p)^{1/p}$.
}
We give an algorithm that is worse than prior works in the sense that it requires a higher value of $d$, but can handle a much larger level of noise than these previous algorithms.
The latter property turns out to be crucial for the dictionary learning application.
Our result for noisy tensor decomposition is captured by the following theorem:

\begin{theorem}[Noisy tensor decomposition] \label{thm:noisy-decomp}
For every $\e>0, \sigma \geq 1$, there exists $d,\tau$ and a
probabilistic $n^{O(\log n)}$-time algorithm $\cR$ such that
for every $\sigma$-dictionary $A=(a^1|\cdots|a^m)$,
given a polynomial $P$ such that
\begin{equation}
\norm{A^{\top}u}_d^d - \tau \norm{u}_2^d \sleq  P \sleq  \norm{A^{\top}u}_d^d + \tau\norm{u}_2^d  \;, \label{eq:tensor-noise}
\end{equation}
$\cR$ outputs with probability at least $0.9$ a set $S$ that is $\e$-close to $\{ a^1,\ldots, a^m \}$.
\end{theorem}

(We denote $P \sleq Q$ if $Q-P$ is a sum of squares of polynomials. Also, as in Theorem~\ref{ithm:dict-learn}, there are certain conditions under which $\cR$ runs in polynomial time; see \nips{the full version of this paper} {Section~\ref{sec:polytime}}.)

The condition (\ref{eq:tensor-noise}) implies that the input $P$ to $\cR$ is  $\tau$-close to the tensor $\norm{A^\top u}_d^d$,
in the sense that $| P(u) - \norm{A^\top u}_d^d| \leq \tau$ for every unit vector $u$.
This allows for very significant noise, since for a typical vector $u$, we expect $\norm{A^\top u}_d^d$ to be have magnitude roughly $mn^{-d/2}$ which would be \emph{much}
smaller than $\tau$ for every constant $\tau>0$.
Thus, on most of its inputs, $P$ can behave radically differently than $\norm{A^\top u}_d^d$, and in particular have many local minima that do not correspond to local minima of the latter
polynomial.
For this reason, it seems unlikely that one can establish a result such as Theorem~\ref{thm:noisy-decomp} using a local search algorithm.\footnote{The conditions (\ref{eq:tensor-noise})
and $\max_{\norm{u}_2=1} | P(u) - \norm{A^\top u}_d^d| \leq \tau$ are not identical for $d>2$.
Nevertheless, the discussion above applies to both conditions, since (\ref{eq:tensor-noise}) does allow for $P$ to have very different behavior than $\norm{A^\top u}_d^d$.}

We give an overview of our algorithm and its analysis in Section~\ref{sec:overview}. \nips{The full paper contains}{Sections~\ref{sec:dict-learn}, \ref{sec:tensor-decomp-proof} and \ref{sec:sample} contain} the complete formal proofs.
In its current form, our algorithm is efficient only in the theoretical/asymptotic sense, but it is very simple to describe (modulo its calls to the SOS solver), see Figure~\ref{fig:basic-alg}.
We believe that the Sum of Squares algorithm can be a very useful tool for attacking machine learning problems, yielding a first solution to the problem that can later be tailored and optimized.

\Bnote{}

\subsection{Related work} \label{sec:relworks}

Starting with the work of Olshausen and Field~\cite{OlshausenF96,OlshausenF96b,OlshausenF97}, there is a vast body of literature
using various heuristics (most commonly alternating
minimization) to learn dictionaries for sparse coding, and applying this tool to many applications.
Here we focus on papers that gave algorithms with \emph{proven} performance.

\emph{Independent Component Analysis (ICA)}~\cite{comon1994independent} is one method that can be used for the dictionary learning
in the case the random variables $x_1,\ldots,x_n$ are  statistically independent.
For the case of $m=n$ this was shown in~\cite{comon1994independent,FriezeJerrumKannan96,NguyenR09}, while the works~\cite{DeLathauwer07,GoyalVX13} extend it for the overcomplete (i.e. $m>n$) case.
\Bnote{}

Another recent line of works analyzed different algorithms, which in some cases are more efficient or handle more general distributions
than ICA. Spielman, Wang and Wright~\cite{SpielmanWW12} give an algorithm to exactly recover the dictionary in the $m=n$ case.
Agarwal, Anandkumar, Jain, Netrapalli, and Tandon~\cite{AgarwalA0NT13}  and Arora, Ge and Moitra~\cite{AroraGM13} obtain
approximate recovery in the overcomplete (i.e. $m>n$) case,  which can be boosted to exact recovery
under some additional conditions on the sparsity and dictionary~\cite{AgarwalAN13,AroraGM13}.
However, all these works require the distribution $x$ to be over \emph{very} sparse vectors, specifically having less than $\sqrt{n}$ nonzero entries.
As discussed in~\cite{SpielmanWW12,AroraGM13}, $\sqrt{n}$ sparsity seemed like a natural barrier for this problem,
and in fact, Spielman et al~\cite{SpielmanWW12} proved that every algorithm of similar nature to theirs will fail to recover the dictionary when when the coefficient vector can have $\Omega(\sqrt{n \log n})$ coordinates.
The only work we know of that can handle vectors of support larger than $\sqrt{n}$ is the recent paper~\cite{AroraBGM14},
but it achieves this at the expense of making fairly strong assumptions on the structure of the dictionary, in particular assuming some sparsity conditions on $A$ itself.
In addition to the sparsity restriction, all these works had additional conditions on the distribution that are incomparable or stronger than ours, and the works~\cite{AgarwalA0NT13,AroraGM13,AgarwalAN13,AroraBGM14}
make additional assumptions on the dictionary (namely incoherence) as well.  \Bnote{}

The tensor decomposition problem is also very widely studied with a long history (see e.g.,~\cite{tucker1966some,harshman1970foundations,kruskal1977three}). Some recent works providing algorithms and analysis include~\cite{AnandkumarFHKL12,AroraGM12,BhaskaraCMV14,BhaskaraCV14}.
However, these works are in a rather different parameter regime than ours--- assuming the tensor is given with very little noise (inverse polynomial in the spectral norm), but on the other hand requiring very low order moments (typically three or four, as opposed to the large constant or even logarithmic number we use).

As described in Sections~\ref{sec:overview} and~\ref{sec:SOS} below, the main tool we use is the \emph{Sum of Squares} (SOS) semidefinite programming hierarchy~\cite{Shor87,Nesterov00,Parrilo00,Lasserre01}.
We invoke the SOS algorithm using the techniques and framework introduced by Barak, Kelner and Steurer~\cite{BarakKS14}.
In addition to introducing this framework, \cite{BarakKS14}  showed how a somewhat similar technical barrier can be bypassed in a setting related to dictionary learning---
the task of recovering a sparse vector that is planted in a random subspace of $\R^n$ given a basis for that subspace.
Assuming the subspace has dimension at most $d$, \cite{BarakKS14} showed that the vector can be recovered as long as it has less than $\min(\e n , n^2/d^2)$ nonzero coordinates for some constant $\e>0$, thus improving (for $d \ll n^{2/3}$) on the prior work \cite{DemanetH13} that required the vector to be $o(n/\sqrt{d})$ sparse.

\nips{}{
\section*{Organization of this paper}

In Section~\ref{sec:overview} we give a high level overview of our ideas. Sections~\ref{sec:dictionary}--\ref{sec:tensor-decomp-proof}
contain the full proof for solving the dictionary learning and tensor decomposition problems in quasipolynomial time, where the sparsity parameter
$\tau$ is a small constant. In Section~\ref{sec:polytime} we show how this can be improved to polynomial time when $\tau \leq n^{-\delta}$ for some constant $\delta>0$.

}

%% file: overview.tex
\section{Overview of algorithm and its analysis} \label{sec:overview}

The dictionary learning problem can be easily reduced to the noisy tensor decomposition problem.
Indeed, it is not too hard to show that for an appropriately chosen parameter $d$, given a sufficiently large number of examples  $y_1,\ldots, y_N$ from the distribution $\{ y = Ax \}$,
the polynomial
\begin{equation}
P= \tfrac{1}{N}\sum_{i=1}^N \iprod{y_i,u}^{2d} \label{eq:empirical-poly}
\end{equation}
will be roughly $\tau$ close (in the spectral norm) to the polynomial $\norm{A^\top u}_d^d$, where $\tau$ is the ``niceness''/``sparsity'' parameter of the distribution $\{ x \}$.
Therefore, if we give $P$ as input to the tensor decomposition algorithm of Theorem~\ref{thm:noisy-decomp}, we will obtain a set that is close to the columns of $A$.\footnote{%
The polynomial (\ref{eq:empirical-poly}) and similar variants have been used before in works on dictionary learning.
\Dnote{}
The crucial difference is that those works made strong assumptions, such as independence of the entries of $\{ x \}$, that ensured this polynomial has a special structure
that made it possible to efficiently optimize over it. In contrast, our work applies in a much more general setting.}

The challenge is that because $\tau$ is a positive constant, no matter how many samples we take, the polynomial $P$ will always be bounded away from the tensor $\norm{A^\top u}_d^d$.
Hence we must use a tensor decomposition algorithm that can handle a very significant amount of noise. This is where the Sum-of-Squares algorithm comes in.
This is a general tool for solving systems of polynomial equations~\cite{Shor87,Nesterov00,Parrilo00,Lasserre01}. Given the SOS algorithm, the description of our tensor decomposition algorithm is
extremely simple (see Figure~\ref{fig:basic-alg} below).  We now describe the basic facts we use about the SOS algorithm, and sketch the analysis of our noisy tensor decomposition algorithm.
See the survey \cite{BarakS14} and the references therein for more detail on the SOS algorithm, and \nips{the full version}{Sections~\ref{sec:dict-learn}, \ref{sec:sample} and \ref{sec:tensor-decomp-proof}}
 for the full description of our algorithm and its analysis (including its variants that take polynomial time at the expense of requiring dictionary learning examples with sparser coefficients).

\subsection{The SOS algorithm}
\label{sec:SOS}

The SOS algorithm is a method, based on semidefinite programming, for solving a system of polynomial equations.
Alas, since this is a non-convex and NP-hard problem, the algorithm doesn't always succeed in producing a solution.
However, it always returns some object, which in some sense can be interpreted as a ``distribution'' $\{ u \}$ over solutions of the system of equations.
It is not an actual distribution, and in particular we cannot sample from $\{ u \}$ and get an individual solution, but we can compute
low order moments of $\{ u \}$. Specifically, we make the following definition:

\begin{definition}[Pseudo-expectations]
Let $\R[u]$ denote the ring of polynomials with real coefficients in variables $u=u_1\ldots u_n$.
Let $\R[u]_{k}$ denote the set of polynomials in $\R[u]$ of degree at most $k$.
A \emph{degree-$k$ pseudoexpectation operator for $\R[u]$} is a linear operator $\cL$ that maps
polynomials in $\R[u]_{k}$ into $\R$ and satisfies that $\cL(1) = 1$ and $\cL(P^2) \geq 0$ for every polynomial $P$ of degree at most $k/2$.
\end{definition}

For every distribution $\cD$ over $\R^n$ and $k\in\N$, the operator $\cL$ defined as $\cL(P)=\E_{\cD} P$ for all $P\in\R[x]$ is degree $k$
pseudo-expectation operator.
We will use notation that naturally extends the notation for actual expectations.
We denote pseudoexpectation operators as $\pE_\cD$, where $\cD$ acts as index to distinguish different operators.
If $\pE_\cD$ is a degree-$k$ pseudoexpectation operator for $\R[u]$, we say that $\cD$ is a \emph{degree-$k$ pseudodistribution} for the indeterminates $u$.
In order to emphasize or change indeterminates, we use the notation $\pE_{v\sim \cD} P(v)$.
In case we have only one pseudodistribution $\cD$ for indeterminates $u$, we denote it by $\{u\}$.
In that case, we also often drop the subscript for the pseudoexpectation and write $\pE P$ or $\pE P(u)$ for $\pE_{\{u\}}P$.

We say that a degree-$k$ pseudodistribution $\{u\}$ satisfies the constraint $\{ P = 0 \}$
if $\pE P(u)Q(u) = 0$ for all $Q$ of degree at most $k - \deg P$.
Note that this is a stronger condition than simply requiring $\pE P(u) = 0$.
We say that $\{ u \}$ satisfies $\{ P \ge  0 \}$  if it satisfies the constraint $\{ P - S = 0 \}$ where $S$ is a sum-of-squares polynomial $S\in\R_{k}[u]$.
It is not hard to see that if $\{ u \}$ was an actual distribution, then these definitions imply that all points in the support of the distribution satisfy the constraints.
We write $P\succeq 0$ to denote that $P$ is a sum of squares of polynomials, and similarly we write $P \sgeq Q$ to denote $P-Q \sgeq 0$.

A degree $k$ pseudo-distribution can be represented by the list of $n^{O(k)}$ values of the expectations
of all monomials of degree up to $k$.
It can also be written as an $n^{O(k)}\times n^{O(k)}$ matrix $M$ whose rows and columns correspond to monomials of degree up to $k/2$;
the condition that $\pE P(u)^2 \geq 0$ translates to the condition that this matrix is positive semidefinite.
The latter observation can be used to prove the fundamental fact about pseudo-distributions, namely that we can efficiently optimize over them.
This is captured in the following theorem:

\begin{theorem}[The SOS Algorithm~\cite{Shor87,Nesterov00,Parrilo00,Lasserre01}] \label{thm:SOS}
For every $\e>0$, $k,n,m,M\in\N$ and $n$-variate polynomials $P_1,\ldots,P_m$ in $\R_k[u]$, whose coefficients are in $\{0,\ldots,M\}$,
if there exists a degree $k$ pseudo-distribbution $\{ u \}$ satisfying the constraint $\{ P_i = 0\}$ for every $i \in [m]$,
then we can find in $(n\polylog(M/\e))^{O(k)}$ time a pseudo-distribution $\{ u' \}$ satisfying  $\{ P_i \leq \e \}$ and $\{ P_i \geq - \e \}$
for every $i\in [m]$.
\end{theorem}

Numerical accuracy will never play an important role in our results, and so we can just assume that we can always find in $n^{O(k)}$
time a degree-$k$ pseudo-distribution satisfying given polynomial constraints, if such a pseudo-distribution exists.

\subsection{Noisy tensor decomposition}

\begin{figure}[t]
\fbox{\begin{minipage}{\textwidth}
\noindent \textbf{Input:} Accuracy parameter $\e$. A degree $d$ polynomial $P$ such that
\[
\norm{A^\top u}_d^d - \tau \norm{u}_2^d \sleq P \sleq  \norm{A^\top u}_d^d + \tau \norm{u}_2^d \,
\]
where $d$ is even.

\medskip

\noindent\textbf{Operation:}
\begin{enumerate}

\item \label{itm:runSOS} Use the SOS algorithm to find the degree-$k$ pseudo-distribution $\{ u \}$ that maximizes $P(u)$ while satisfying $\norm{u}^2 \equiv 1$.

\item \label{itm:pickW} Pick the polynomial $W$ to be a product of $O(\log n)$ random linear functions.

\item \label{itm:output}  Output the top eigenvector of the matrix $M$ where $M_{i,j} = \pE W(u)^2 u_iu_j$.

\end{enumerate}
\end{minipage}}
\caption{Basic Tensor Decomposition algorithm. \footnotesize{The parameters $k,d,\tau$ are chosen as a function of the accuracy parameter $\e$ and
the top eigenvalue $\sigma$ of $A^\top A$. The algorithm outputs a vector $u$ that is $\e$-close to a column of $A$
with inverse polynomial probability.}}
\label{fig:basic-alg}
\end{figure}

Our basic noisy tensor decomposition algorithm is described in Figure~\ref{fig:basic-alg}.
This algorithm finds  (a vector close to) a column of $A$ with inverse polynomial probability.
Using similar ideas, one can extend it to an algorithm that outputs all vectors with high probability; we provide the details in \nips{the full version.}{Section~\ref{sec:tensor-decomp-proof}.}
Following the approach of \cite{BarakKS14}, our analysis of this algorithm proceeds in two phases:

\begin{description}

\item[(i)] We show that if the pseudo-distribution $\{ u \}$  obtained in Step~\ref{itm:runSOS} is an \emph{actual}
distribution, then the vector output in Step~\ref{itm:output} is close to one of the columns of $A$.

\item[(ii)] We then show that the arguments used in establishing \textbf{(i)} generalize to the case
of pseudo-distributions as well.

\end{description}

\paragraph{Part \textbf{(i)}} The first part is actually not so surprising.
For starters, every unit vector $u$ that maximizes $P$ must be highly correlated with some column $a$ of $A$.
Indeed,  $\norm{A^\top a}_d^d \geq 1$ for every column $a$ of $A$, and hence the maximum of $P(u)$ over a unit $u$
is at least $1-\tau$. But if $\iprod{u,a}^2 \leq 1-\e$ for every column $a$ then  $P(u)$ must be much smaller than $1$.
Indeed, in this case
\begin{equation}
\norm{A^\top u}_d^d = \sum_i \iprod{a^i,u}^d  \leq \max_i \iprod{a^i,u}^{d-2} \sum \iprod{a^i,u}^2 \;. \label{eq:gsdfgjhk}
\end{equation}
Since $\sum \iprod{a^i,u}^2 \leq \sqrt{\sigma}$, this implies that, as long as $d \gg \tfrac{\log \sigma}{\e}$,
 $\norm{A^\top u}_d^d$ (and thus also $P(u)$) is much smaller than $1$.

Therefore, if $\{ u \}$ obtained in Step~\ref{itm:runSOS} is an actual distribution, then it would be  essentially supported on the set
$\cA = \{ \pm a^1 , \ldots, \pm a^m \}$ of the columns of $A$ and their negations.
Let us suppose that $\{ u \}$ is simply the uniform distribution over $\cA$.
(It can be shown that this essentially is the hardest case to tackle.)
In this case the matrix $M$ considered in Step~\ref{itm:output} can be written as
\[
M = \tfrac{1}{m} \sum_{i=1}^m W(a^i)^2  (a^i)(a^i)^\top  \;,
\]
where $W(\cdot)$ is the polynomial selected in Step~\ref{itm:pickW}.
(This uses the fact that this polynomial is a product of linear functions and hence satisfies $W(-a)^2 = W(a)$ for all $a$.)
If $W(\cdot)$ satisfies
\begin{equation}
|W(a^1)| \gg \sqrt{m} |W(a^i)|  \label{eq:Wisolates}
\end{equation}
for all $i\neq 1$ then $M$ is very  close to (a constant times) the matrix $(a^1)(a^1)^\top$, and hence its top eigenvector is close
to $a^1$ and we would be done.
We want to show that the event (\ref{eq:Wisolates}) happens with probability at least inverse polynomial in $m$.
Recall that $W$ is a product of $t=c\log n$ random linear functions for some constant $c$ (e.g., $c=100$ will do).
That is, $W(u) = \prod_{i=1}^t \iprod{v^i,u}$, where $v^1,\ldots,v^t$ are standard random Gaussian vectors.
Since $\E \iprod{v^j,a^i}^2 =1$ and these choices are independent, $\E W(a^i)^2=1$  for all $i$.
However, with probability $\exp(-O(t)) = m^{-O(1)}$ it will hold that $|\iprod{v^j,a^1}| \geq 2$ for all $j=1\ldots t$.
In this case $|W(a^1)|\geq  2^t$, while we can show that even conditioned on this event, with high probability we
will have $|W(a^i)| < 1.9^t \ll |W(a^1)|/\sqrt{m}$ for all $i$, in which case (\ref{eq:Wisolates}) holds.\footnote{%
This argument assumes that no other column is $0.9$ correlated with $a^1$.
However our actual analysis does not use this assumption,
since if two column vectors are closely correlated, we are fine with outputting any linear combination of them.}

\textbf{Part (ii).} The above argument establishes \textbf{(i)}, but this is all based on a rather bold piece of wishful thinking---
that the object $\{ u \}$ we obtained in Step~\ref{itm:runSOS} of the algorithm
was actually a genuine distribution over unit vectors maximizing $P$.
In actuality, we can only obtain the much weaker guarantee that $\{ u \}$ is a degree $k$ \emph{pseudo-distribution} for some $k= O(\log n)$.
(An actual distribution corresponds to a degree-$\infty$ pseudo-distribution.)
The technical novelty of our work lies in establishing \textbf{(ii)}.
The key observation is that in all our arguments above, we never used any higher moments of $\{ u \}$, and that all the inequalities we showed
boil down to the simple fact that a square of a polynomial is never negative.
(Such proofs are known as \emph{Sum of Squares (SOS) proofs}.)

We will not give the full analysis here, but merely show a representative example of how one ``lifts'' arguments into the SOS
setting. In (\ref{eq:gsdfgjhk}) above we used the simple inequality that for every vector $v\in \R^m$
\begin{equation}
\norm{v}_d^d \leq \norm{v}^{d-2}_{\infty} \norm{v}_2^2 \;, \label{eq:lkgfsdjglkfs}
\end{equation}
applying it to the vector $v= A^\top u$ (where we denote $\norm{v}_\infty = \max_i |v_i|$).
The first (and most major) obstacle in giving a low degree ``Sum of Squares'' proof for (\ref{eq:lkgfsdjglkfs}) is that this is not a
polynomial inequality.
To turn it into one, we replace the $L_{\infty}$ norm with the $L_k$ norm for some large $k$ ($k=O(\log m)$ will do).
If we replace $\norm{v}_\infty$ with $\norm{v}_k$ in (\ref{eq:lkgfsdjglkfs}), and raise it to the $k/(d-2)$-th  power then we obtain the inequality
\begin{equation}
\left( \norm{v}_d^d\right)^{k/(d-2)} \leq \norm{v}_k^k \left(\norm{v}_2^2\right)^{k/(d-2)} \;, \label{eq:norm-inequality}
\end{equation}
which is a valid inequality between polynomials in $v$ whenever $k$ is an integer multiple of $d-2$ (which we can ensure).

We now need to find a sum-of-squares proof for this inequality, namely that the right-hand side of (\ref{eq:norm-inequality}) is equal to the left-hand side plus a sum of squares, that is, we are to show that for $s=k/(d-2)$,
$$	\bigg(\sum_i v_i^d\bigg)^s
	\sleq
	\bigg(\sum_i v_i^{(d-2)s}\bigg)
	\bigg(\sum_i v_i^2\bigg)^s.
$$
By expanding the $s$-th powers in this expression, we rewrite this polynomial inequality as
\begin{equation}\label{eq:norm-inequality-expanded}
	\sum_{\lvert  \alpha  \rvert=s}\binom{s}{\alpha}v^{d\alpha}
	\sleq
	\bigg(\sum_i v_i^{(d-2)s}\bigg) \sum_{|\alpha|=s}\binom{s}{\alpha}v^{2\alpha}
	=
 \sum_{|\alpha|=s}\binom{s}{\alpha}\,v^{2\alpha}	\sum_i v_i^{(d-2)s},
\end{equation}
where the summations involving $\alpha$ are over degree-$s$ multiindices $\alpha\in\{0,\dots,s\}^n$, and $\binom{s}{\alpha}$ denotes the multinomial coefficient $\binom{n}{\alpha}=\frac{s!}{\alpha_1!\dots\alpha_m!}$.
We will prove \pref{eq:norm-inequality-expanded} term by term, i.e., we will show that $v^{d\alpha}\preceq  v^{2\alpha}\sum_i v_i^{(d-2)s}$ for every multiindex $\alpha$.
Since $v^{2\alpha}\succeq 0$, it is enough to show that $v^{(d-2)\alpha}\preceq \sum_i v_i^{(d-2)s}$.
This is implied by the following general inequality\nips{ (proven in the full version of this paper)}{,
 which we prove in Appendix~\ref{sec:monomial-ineq}}:

\begin{lemma}
  \label{lem:monomial-ineq}
  Let $w_1,\ldots,w_n$ be polynomials.
  Suppose $w_1\succeq 0,\ldots,w_n\succeq 0$.
  Then, for every multiindex $\alpha$,
  \nips{$$w^\alpha \preceq \sum_i w_i^{\lvert  \alpha \rvert}\mper$$}{$w^\alpha \preceq \sum_i w_i^{\lvert  \alpha \rvert}\mper$}
\end{lemma}
\noindent We note that $d$ is even, so $w_i=v_i^{d-2}\succeq 0$ is a square, as required by  the lemma.

For the case that $\lvert  \alpha \rvert$ is a power of $2$, the inequality in the lemma follows by repeatedly applying the inequality $x\cdot y\preceq \frac 12 x^2+ \frac 12 y^2$, which in turn holds because the difference between the two sides equals $\tfrac 12(x-y)^2$.
As a concrete example, we can derive $w_1^3w_2\preceq w_1^4 + w_2^4$ in this way,
\begin{displaymath}
  w_1^3w_2=w_1^2\cdot w_1w_2
  \preceq \tfrac12 w_1^4 + \tfrac12 w_1^2\cdot w_2^2
  \preceq  \tfrac12 w_1^4 + \tfrac12 \bigl(\tfrac12 w_1^4 + \tfrac12 w_2^4\bigr)
  \preceq w_1^4 + w_2^4\mper
\end{displaymath}
(The first two steps use the inequality $x\cdot y\preceq \tfrac12 x^2+\tfrac12 y^2$.
The last step uses that both $w_1$ and $w_2$ are sum of squares.)

\Dcomment{comment out previous exposition}

\Dnote{}

Once we have an SOS proof for (\ref{eq:norm-inequality}) we can conclude that it holds for pseudo-distributions as well, and in particular that for
every pseudo-distribution  $\{ u \}$ of degree at least $k+ 2k/(d-2)$ satisfying  $\{ \norm{u}_2^2 = 1\}$,
\stoceq{
\pE \left( \norm{A^\top u}_d^d\right)^{k/(d-2)} \leq \pE \norm{A^\top u}_k^k \sigma^{k/(d-2)} \;.
}
We use similar ideas to port the rest of the proof to the SOS setting, concluding that whenever $\{ u \}$ is a pseudo-distribution
that satisfies $\{ \norm{u}_2^2 = 1 \}$ and $\{ P(u)\geq 1-\tau \}$, then with inverse polynomial probability it will hold that
\begin{equation}
\pE W(u)^2 \iprod{u,a}^2 \geq (1-\e) \pE W^2  \label{eq:gfhkjgdf}
\end{equation}
for some column $a$ of $A$ and $\e>0$ that can be made arbitrarily close to $0$.
Once we have (\ref{eq:gfhkjgdf}), it is not hard to show that the matrix $M=\pE W(u)^2 uu^\top$  obtained in Step~\ref{itm:output}
of our algorithm  is close to $aa^\top$.
Hence, we can recover a vector close to $\pm a$ by computing the top eigenvector\footnotemark\ of the matrix $M$.
\footnotetext{
In the final algorithm, instead of computing the top eigenvector of the matrix $M$, we will sample from a Gaussian distribution $\{\xi \}$ that satisfies $\E \xi \xi ^\top =M$.
If $M\approx a a^\top$, then such a Gaussian vector $\xi$ is close to $\pm a$ with high probability.}

%% file: prelim.tex
\section{Preliminaries} \label{sec:prelims}

\Bnote{}

\nips{}{We recall some of the notation mentioned above. We use $P \sleq Q$ to denote that $Q-P$ is a sum of square polynomials.
For a vector $v\in \R^d$ and $p\geq 1$, we denote $\norm{v}_p = (\sum_{i=1}^d |v_i|^p)^{1/p}$ and $\norm{v} = \norm{v}_2$.
For any $\sigma \geq 1$, a \emph{$\sigma$-dictionary} is an $n\times m$ matrix $A = (a^1|\cdots|a^m)$ such that $\norm{a^i}=1$ for all $i$ and the spectral norm of $A^\top A$ is
at most $\sigma$ or ,equivalently,  $\norm{Au}_2^2 \sleq \sigma \norm{u}_2^2$.
Two sets $S_0,S_1 \subseteq \R^n$ are \emph{$\e$-close in symmetrized Hausdorff distance} if for all $b\in 0,1$, $\min_{s\in S_b} \max_{t \in S_{1-b}} \cor(s,t) \geq 1-\e$, where $\cor(s,t) =\iprod{s,t}^2/(\norm{s}\norm{t})^2$; we often drop the qualifier ``symmetrized Hausdorff distance'' as we will not use another notion of distance between sets of vectors in this paper.}

We use the notation of  pseudo-expectations and pseudo-distributions from \nips{Section 2.1}{\pref{sec:SOS}}.
We now state some basic useful facts about pseudo-distributions,  see~\cite{BarakS14,BarakKS14,BarakBHKSZ12} for a more comprehensive treatment.

One useful property of pseudo-distributions is that we can find actual distribution that match their first two moments.

\begin{lemma}[Matching first two moments] \label{lem:moments}
Let $\{ u \}$ be a pseudo-distribution over $\R^n$ of degree at least $2$.
Then we can efficiently sample from a Gaussian distribution\footnote{A Gaussian distribution with covariance $\Sigma\in\R^{n\times n}$ and mean $\mu\in\R^n$ has density proportional to $x\mapsto\exp(-\langle x-\mu,\Sigma^{-1}(x-\mu) \rangle/2)$. } $\{\xi\}$ over $\R^n$ such that for every polynomial $Q$ of degree at most $2$,
$$\E Q(\xi) = \pE Q(u).$$
\end{lemma}
\Dnote{}
\begin{proof}
By shifting, it suffices to restrict attention to the case where $\E u_i = 0$ for all $i$.
Consider the matrix $M$ such that $M = \pE uu^\top$.
The positivity condition implies that $M$ is a positive semidefinite matrix.
Therefore, $M$ admits a Cholesky factorization $M=VV^\top$.
Let $\{\zeta \}$ be the standard Gaussian distribution on $\R^n$ (mean $0$ and variance $1$ in each coordinate) and consider the Gaussian distribution $\{\xi = V \zeta\}$.
We are to show that $\xi$ has the same degree-$2$ moments as the pseudo-distribution $\{u\}$.
Indeed,
\begin{displaymath}
  \E \xi \xi^\top = \E V   \zeta \zeta^\top V^\top = V V^\top = M= \pE u u^\top \mper
\end{displaymath}
Here, we use that $\E \zeta \zeta ^\top$ is the identity because $\zeta$ is a standard Gaussian vector.
\end{proof}

Another property we will use is that we can \emph{reweigh} a pseudo-distribution by a positive polynomial $W$ to obtain a new pseudo-distribution
that corresponds to the operation on actual distributions of reweighing the probability of an element $u$ proportional to $W(u)$.

\begin{lemma}[Reweighing] \label{lem:reweigh} Let $\{ u \}$ be a degree-$k$ pseudo-distribution.
Then for every SOS polynomial $W$ of degree $d<k$ with $\pE W > 0$, there exists a degree-$(k-d)$ pseudo-distribution $\{ u' \}$ such that  for every polynomial $P$ of degree at most $k-d$
\begin{displaymath}
  \pE_{\{u'\}} P(u') = \tfrac1{\pE_{\{u\}} W(u)}\pE_{\{u\}} W(u)P(u)\mper
\end{displaymath}
\end{lemma}
\begin{proof}
The functional $\pE_{\{u'\}}$ is linear and satisfies $\pE_{\{u'\}} 1 =1$, and so we just need to verify the positivity property.
For every polynomial $P$ of degree at most $(k-\deg W)/2$,
\[
\pE_{\{u'\}} P(u')^2 = (\pE_{\{u\}} W(u)P(u)^2 ) / (\pE_{\{u\}} W(u) )
\]
but since $W$ is a sum of squares, $WP^2$ is also a sum of squares and hence the denominator of the left-hand side is non-negative, while the numerator is by assumption positive.
\end{proof}

%% file: dictionary.tex
\section{Dictionary Learning}
\label{sec:dictionary}
\label{sec:dict-learn}

We now state our formal theorem for dictionary learning.
The following definition of nice distributions captures formally the conditions needed for recovery.
(It is equivalent up to constants to the definition of \nips{Section 1.1}{\pref{sec:intro-nice}}, see \pref{rem:niceness-definition} below.)

\begin{definition}[Nice distribution] \label{def:nice}
Let $\tau \in (0,1)$ and $d \in \N$ with $d$ even.
A distribution $\{ x \}$ over $\R^m$ is $(d,\tau)$-nice if it satisfies the following properties:
\begin{enumerate}
\item \label{itm:nice-lowerbound}  $\E x_i^d = 1$ for all $i\in[m]$,
\item \label{itm:nice-evensmall}  $\E x^\alpha \le \tau $ for all degree-$d$ monomials $x^\alpha \notin \{x_1^d,\ldots,x_m^d\}$, and
\item \label{itm:nice-oddvanish} $\E x^\alpha =0$ for all non-square degree-$d$ monomials $x^\alpha$.
\end{enumerate}
Here, $x^\alpha$ denotes the monomial $x_1^{\alpha_1}\cdots x_m^{\alpha_m}$.
Furthermore, we require that $x_i^d$ to have polynomial variance so that $\E x_i^{2d}=n^{O(1)}$.
To avoid some technical issues, $(d,\tau)$-nice distributions are also assumed to be $(d',\tau)$-nice after rescaling for all even $d'\le d$.
Concretely, when we say that $\{x\}$ is a $(d,\tau)$-nice distribution, we also imply that for every positive even $d'< d$, there exists a rescaling factor $\lambda$ such that the distribution $\{\lambda \cdot x\}$ satisfies the three properties above (plus polynomial variance bound).
\end{definition}
\Dnote{}

Let us briefly discuss the meaning of these conditions.
The condition $\E x_i^{2d}=n^{O(1)}$ is a weak non-degeneracy condition, ruling out distributions where the main contribution to some low order moments comes from events that happen with super-polynomially small probability.
Condition~\ref{itm:nice-lowerbound} stipulates that we are in the symmetric case, where all coefficients have more or less the same magnitude.
(We can remove symmetry by either dropping this condition or allowing the dictionary vectors to have different norms; see \pref{rem:different-norms}.)
\Dnote{}
\Dnote{}
Condition~\ref{itm:nice-evensmall} captures to a certain extent both the sparsity conditions and that that the random variables $x_i$ and $x_i$ for $i\neq j$ are not too correlated.
Condition~\ref{itm:nice-oddvanish} stipulates that there is significant ``cancellations'' between the negative and positive coefficients.
While it is satisfied by many natural distributions, it would be good to either show that it can be dropped, or that it is inherently necessary.
The requirement of having expectation zero---perfect cancellation---can  be somewhat relaxed to having a sufficiently small bound (inverse polynomial in $n$) on the magnitude of the non-square moments.

We can now state our result for dictionary learning in quasipolynomial time.
The result for polynomial time is stated in Section~\ref{sec:polytime}.

\begin{theorem}[Dictionary learning, quasipolynomial time]
\label{thm:dict-learn}
There exists an algorithm that for every desired accuracy $\e>0$ and overcompleteness $\sigma\ge 1$ solves the following problem \Dfnote{}
for every $(d,\tau)$-nice distribution with $d \ge d(\e,\sigma)=O(\e^{-1}\log \sigma )$ and $\tau\le \tau(\e,\sigma)=(\e^{-1}\log \sigma )^{O(\e^{-1}\log \sigma )}$ in time $n^{(1/\e)^{O(1)}  (d+\log m) }$:
Given $n^{O(d)}/\poly(\tau)$ samples from a distribution $\{y=Ax\}$ for a $\sigma$-overcomplete dictionary $A$ and $(d,\tau)$-nice distribution $\{x\}$, output a set of vectors that is $\e$-close to the set of columns of $A$ (in symmetrized Hausdorff distance).

\end{theorem}

In the \emph{tensor decomposition problem}, we are given a polynomial of the form $\norm{A^\top u}_d^d\in\R[u]$ (or equivalently a tensor of the form $\sum_i a_i^{\otimes d}$) and our goal is to recover the vectors $a_1,\ldots,a_m$ (up to signs).
It turns out the heart of the dictionary learning problem is solving a variant of the tensor decomposition problem, where we are not given the polynomial $\norm{A^\top u}_d^d$ but a polynomial close to it in spectral norm.
(The magnitude of this error is related to the niceness of the distribution, which means that we cannot assume it to be arbitrarily small.)

\begin{theorem}[Noisy tensor decomposition]
\torestate{
\label{thm:tensor-decomp}
There exists an algorithm that for every desired accuracy $\e>0$ and overcompleteness $\sigma\ge 1$ solves the following problem  for every degree $d \ge d(\e,\sigma)=O(1/\e)\cdot \log \sigma $ and noise parameter $\tau\le \tau(\e)=\Omega(\e)$ in time $n^{(1/\e)^{O(1)} (d + \log m)}$:
Given a degree-$d$ polynomial $P\in\R[u]$ that is $\tau$-close to $\lVert  A^\top u \rVert_d^d$ in spectral norm for a $\sigma$-overcomplete dictionary $A$, i.e.,
\[
\norm{A^\top u}_d^d + \tau \norm{u}_2^d  \sgeq P(u) \sgeq  \norm{A^\top u}_d^d - \tau \norm{u}_2^d \;,
\]
output a set of vectors that is $\e$-close to the set of columns of $A$ (in symmetrized Hausdorff distance).
}
\end{theorem}

\begin{remark}[Different notions of niceness]
\label{rem:niceness-definition}
In \nips{Section~1.1}{\pref{sec:intro-nice}} we defined $(d,\tau)$-niceness in a different way.
Instead of requiring $\E x^\alpha \le \tau$ for every monomial $x^\alpha\notin \{x_1^d,\ldots,x_m^d\}$, we only required this condition for some of these monomials, namely monomials of the form $x^\alpha=x_{i}^{d/2} x_j^{d/2}$.
It turns out that these two definitions are equivalent up to a factor $d$ in the exponent of $\tau$.
(This loss of a factor of $d$ in the exponent is OK, since in our applications $\tau$ will anyway be exponentially small in $d$.)
To see the equivalence of the definitions, note that every degree-$d$ square monomial $x^\alpha \notin \{x_1^d,\ldots,x_m^d\}$ involves at least two distinct variables, say $x_i$ and $x_j$, and therefore
\begin{math}
  x^\alpha = \E x_i^2x_j^2x^{\alpha'}
  \mcom
\end{math}
where $x^{\alpha'}$ is a monomial of degree $d-4$ (so that $\sum_k \alpha_k' = d-4$).
By \Hoelder's Inequality, we can bound its expectation
\[
\E x_i^2x_j^2 x^{\alpha'} \le  \left( \E x_i^{d/2} x_j^{d/2} \right)^{4/d} \left( \E x^{\beta}\right)^{(d-4)/d} \;,
\]
for $\beta = \tfrac{d}{d-4}\alpha'$.
Since $\sum \beta_k = d$, the Arithmetic-Mean Geometric-Mean Inequality together with our normalization $\E x_k^d=1$ implies
\[
\E x^{\beta} \leq \sum_k \tfrac {\beta_k } d \cdot  \E x_k^d = 1 \;,
\]
thus proving that $\E x^{\alpha} \leq ( \E x_i^{d/2} x_j^{d/2} )^{4/d}$ for every degree-$d$ square monomial $x^\alpha\notin \{x_1^d,\ldots,x_m^d\}$.
\end{remark}

\subsection{Dictionary learning via noisy tensor decompostion}

We will prove \pref{thm:tensor-decomp} (noisy tensor decomposition) in \pref{sec:sample} and \pref{sec:tensor-decomp-proof}.
At this point, let us see how it yields \pref{thm:dict-learn} (dictionary learning, quasipolynomial time).
The following lemma gives the connection between tensor decomposition and dictionary learning.

\begin{lemma}\label{lem:reduce-tensor}
  Let $\{x\}$ be a $(d,\tau)$-nice distribution over $\R^m$ and $A$ a $\sigma$-overcomplete dictionary.
  Then,\footnote{The factor $d^d$ can be somewhat reduced, e.g., to  $d^{d/2}$.
  However, this improvement would be hidden by $O(\cdot)$ notation at a later point.
  For simplicity, we will work with the simple $d^d$ bound at this point.
}
  \[
  \lVert  A^\top u \rVert_d^d + \tau \sigma^d d^d \lVert  u \rVert_2^d \succeq \E_x \langle  Ax,u \rangle^d  \succeq  \lVert  A^\top u \rVert_d^d \;.
  \]
\end{lemma}

\begin{proof}
Consider the polynomial $p(v) = \lVert  v \rVert_d^d + \tau d^d \lVert  v \rVert_2^d - \E_x \langle  x,v \rangle^d$ in the monomial basis for the variables $v_1,\ldots,v_m$.
All coefficients corresponding to non-squared monomials are zero (by the third property of nice  istibutions).
All other coefficients are nonnegative (by the first and second property of nice distributions).
We conclude that $p$ is a sum of squares.
The relation $\lVert  A^\top u \rVert_d^d + \tau \sigma d^d \lVert  u \rVert_2^d \succeq \E_x \langle  Ax,u \rangle^d$ follows by substituting $v=A^\top u$ and using the relation $\lVert  A^\top u \rVert_2^d \preceq \sigma^d \lVert  u \rVert_2^d$.

For the lower bound, we see that the polynomial $q(v)=\E_x \langle  x,v \rangle^d - \lVert  v \rVert_d^d$ is a nonnegative combination of square monomials.
Thus, $q(v)\sge 0$ and the desired bound follows by substituting $v=A^\top u$.
\qedhere
\end{proof}

\begin{proof}[Proof of \pref{thm:dict-learn}]
If we take a sufficiently large number of samples $y_1,\ldots,y_N$ from the distribution $\{ y = Ax \}$
(e.g., $N \geq n^{O(d)}/\tau^2$ will do), then with high probability every coefficient of the polynomial $P = \tfrac{1}{N} \sum \iprod{y_i,u}^d\in\R[u]$ would be $\tau/n^d$-close to the corresponding coefficient of $\E \iprod{y,u}^d$.
Therefore, $\pm (P-\E\langle Ax,u\rangle^d)\preceq \tau \cdot \lVert  u \rVert_2^d$.
\Dnote{}
Together with \pref{lem:reduce-tensor} it follows that
\[
\lVert  A^\top u \rVert_d^d + 2\tau \sigma ^d d^d \lVert  u \rVert_2^d \succeq P \succeq  \lVert  A^\top u \rVert_d^d - 2 \tau \sigma^d d^d \lVert  u \rVert_2^d \;.
\]
Therefore, we can apply the algorithm in \pref{thm:tensor-decomp} (noisy tensor decomposition) for noise parameter $\tau'=2\tau k^d d^d$ to obtain a set $S$ of unit vectors that is $\e$-close to the set of columns of $A$ (in symmetrized Hausdorff distance).
\Dnote{}
\Dnote{}
\end{proof}

%% file: sampling.tex
\section{Sampling pseudo-distributions} \label{sec:sample}

In this section we will develop an efficient algorithm that behaves in certain ways like a hypothetical sampling procedure for low-degree pseudo-distributions.
(Sampling procedures, even inefficient or approximate ones, cannot exist in general for low-degree pseudo-distributions \cite{Grigoriev01,Schoenebeck08}.)
This algorithm will be a key ingredient of our algorithm for \pref{thm:tensor-decomp} (noisy tensor decomposition, quasipolynomial time).

Here is the property of a sampling procedure that our algorithm mimics:
Suppose we have a probability distribution $\{u\}$ over unit vectors in $\R^m$ that satisfies $\E\langle  c,u \rangle^k\ge e^{-\e k}$ for some unit vector $c\in\R^m$, small $\e>0$, and $k\gg 1/\e$ (so that $e^{-\e k}$ is very small).
This condition implies that if we sample a vector~$u$ from the distribution then with probability at least $e^{-\e k}/2$ the vector satisfies $\langle c,u  \rangle^k \ge e^{-\e k}/2$, which means $\langle  c,u \rangle^2\ge e^{-2 \e }/2^{-1/k}\ge 1-O(\e)$.
(Since $e^{-\e k}$ was very small to begin with, the additional factor $2$ for the correlation and the probability is insubstantial.)

The algorithm in the following theorem achieves the above property of sampling procedures with the key advantage that it applies to any low-degree pseudo-distributions.

\begin{theorem}[Sampling pseudo-distributions]
  \label{thm:sample}
  For every even $k\ge 0$, there exists a randomized algorithm with running time $n^{O(k)}$ and success probability $2^{-k/\poly(\e)}$ for the following problem:
  Given a degree-$k$ pseudo distribution $\{u\}$ over $\R^n$ that satisfies the polynomial constraint $\lVert  u \rVert_2^2=1$ and the condition $\pE \langle  c,u \rangle^k\ge e^{-\e k}$ for some unit vector $c\in\R^n$, output a unit vector $c'\in\R^n$ with $\langle c,c'  \rangle\ge 1-O(\e)$.
\end{theorem}

The result  follows from the following lemmas.
\Dnote{}

\begin{lemma}
\label{lem:sample:reweigh}
  Let $c\in\R^n$ be a unit vector and let $\{u\}$ be a degree-$(k+2)$ pseudo-distribution over $\R^n$ that satisfies the polynomial constraint $\lVert  u \rVert_2^2=1$.
  Suppose $\pE \langle  c,u \rangle^k\ge e^{-\e k}$ for $\e>0$.
  Then, there exists a degree-$k$ sum-of-squares polynomial $W$ such that
  \begin{displaymath}
    \pE W\cdot \langle c,u \rangle^2 \ge \left (1-O\left(\e\right) \right) \pE W\mper
  \end{displaymath}
  Furthermore, there exists a randomized algorithm that runs in time $n^{O(k)}$ and computes such a polynomial $W$ with probability at least $2^{-O(k/\poly(\e))}$.
\end{lemma}

\begin{proof}

Let us first analyze the random polynomial $w=\langle  \xi,u \rangle^2$ for an $n$-dimensional standard Gaussian vector $\xi$.
Let $\tau_M$ be such that a standard Gaussian variable $\xi_0$ conditioned on $\xi_0\ge \tau_M$ has expectation $\E_{\xi_0\ge \tau_M}\xi_0^2=M$.
This threshold satisfies $\tau_M\le M$ and thus $\Pr\{\xi_0\ge \tau_M\}\ge 2^{-O(M^2)}$.
Conditioned on the event $\langle c,\xi \rangle\ge \tau_{M+1}$, the expectation of the random polynomial $w$ satisfies
\begin{displaymath}
  \E_{\{\xi\,\mid\, \langle  c,\xi \rangle\ge \tau_{M+1}\}} w
  = (M+1) \cdot \langle c, u\rangle ^2  + \lVert  u \rVert_2^2 -\langle  c, u \rangle^2
  = M\cdot \langle  c,u \rangle^2 + \lVert  u \rVert_2^2\mper
\end{displaymath}
(Here, we use that $\xi=\langle  c,\xi \rangle c + \xi'$, where $\xi'$ is a standard Gaussian vector in the subspace orthogonal to $c$ so that $\E\langle  \xi',u \rangle^2=\lVert  u \rVert_2^2-\langle  c,u \rangle^2$.)
\Dnote{}

Let $\super w 1,\ldots,\super w {k/2}$ be independent samples from the distribution $\{w\mid \langle  c,\xi \rangle\ge \tau_{M+1}\}$.
Then, let $W=\super w 1 \cdots \super w {k/2} / M^{k/2}$.
The expectation of this random polynomial satisfies
\begin{displaymath}
  \E W = \left ( \langle  c,u \rangle^2 + \tfrac 1{M}\cdot \lVert  u \rVert_2^2  \right )^{k/2}\mper
\end{displaymath}
Let $\overline W=(\langle c,u \rangle^2 + \ffrac 1{M})^{k/2}$.
Since the pseudo-distribution $\{u\}$ satisfies the constraint $\lVert  u \rVert_2^2=1$, it also satisfies the constraint $\E W=\overline W$.
We claim that,
\begin{equation}\label{eq:1}
  \overline W \cdot \langle c, u \rangle ^2 \succeq \left(1-\tfrac 2{M}\right) \cdot \overline W - \left(1- \tfrac 1{M} \right)^{k/2}\mper
\end{equation}
Consider the univariate polynomial
\begin{displaymath}
  p(\alpha) = \alpha^2\cdot (\alpha^2+\tfrac 1 {M})^{k/2}+(1-\tfrac 1{M})^{k/2} - (1-\tfrac 2 M)(\alpha^2+\tfrac 1 {M})^{k/2}\mper
\end{displaymath}
This polynomial is nonnegative on $\R$, because for $\alpha^2\ge 1-\ffrac 2 M$, the first term cancels the last term, and for $\alpha^2<1-\ffrac 2M$, the second term cancels the last term.
Since $p$ is univariate and nonnegative on $\R$, it follows that $p$ is a sum of squares.
Hence, equation \pref{eq:1} follows by substituting $\alpha=\langle  c,u \rangle$.

The following bound shows that there exists a polynomial $W$ that satisfies the conclusion of the lemma,
\begin{align}
  \E_W \pE W \cdot \langle  c,u \rangle^2 & \ge (1-\tfrac 2M)\E_W \pE W - e^{-k/2M}\notag\\
  &  \ge (1- \tfrac 2M - e^{-1/2\e M}) \E_W \pE W\notag\\
  & \ge \left(1- O\left( \e \right) \right) \E_W \pE W\mper\label{eq:2}
\end{align}
The first step uses \pref{eq:1} and the bound $(1-1/M)\le e^{-1/M}$.
The second step uses that $\E_W \pE W =\pE \overline W \ge \pE \langle  c,u \rangle^k \ge e^{-\e k}$ (premise of the lemma).
For the third step, we choose $M=(1/\e) \cdot \log(1/\e)$ to trade-off the two error terms $2/M$ and $e^{-1/2\e M}$.
\Dnote{}

To show the second part of the lemma, we give a randomized algorithm that runs in time $n^{O(k)}$ and computes a polynomial $W_0$ with the desired properties with probability $2^{-O(k/\poly(\e))}$.
The algorithm samples independent standard Gaussian vectors $\super \xi1,\ldots,\super \xi{k/2}$ and outputs the polynomial $W_0=\tfrac 1{M^{K/2}} \langle  \super \xi 1, u \rangle^2\cdots\langle  \super \xi {k/2},u \rangle^2$.
We are to show that $\pE W_0 \langle  c,u \rangle^2\ge (1-O(\e))\pE W_0$ with probability $2^{-O(k/\poly(\e))}$ over the choice of $W_0$.
The distribution $\{W\}$ has density $2^{-O(M^2)}$ in the distribution $\{W_0\}$, in the sense that there exists an event $\cE$ with $\Pr_{\{W_0\}} \cE\ge 2^{-O(M^2)}$ and $\{W\}=\{W_0\mid \cE\}$.
(The event is $\cE = \{\langle  \super \xi 1,c \rangle,\ldots, \langle  \super \xi {k/2},c \rangle\ge \tau_{M+1}\}$).

We will first bound the second moment  $\E_{W}(\pE W)^2$.
The main step is the following bound on the expectation of the random polynomial $w(u)w(u') \in \R[u,u']_4$,
\begin{align}
  \E_{\{\xi\mid \langle  c,\xi \rangle\ge \tau_{M+1}\}}  w(u) w(u')
  & =   \E_{\{\xi\mid \langle  c,\xi \rangle\ge \tau_{M+1}\}}  \bigl(\langle  c,\xi \rangle \langle  c,u \rangle + \langle  \xi',u \rangle\bigr)^2  \bigl(\langle  c,\xi \rangle \langle  c,u' \rangle + \langle  \xi',u' \rangle\bigr)^2\notag\\
  &  \preceq  2^{100M^2} \bigl(\langle  c,u \rangle^2 + \tfrac 1 {M} \lVert  u \rVert\bigr)^2 \bigl(\langle  c,u' \rangle^2 + \tfrac 1 {M} \lVert  u' \rVert\bigr)^2\label{eq:3}
\end{align}
In the second step, $\xi'$ is a standard Gaussian vector in the subspace orthogonal to $c$.
The third step uses the crude upper bound $\E_{\xi\mid \langle  c,\xi \rangle\ge \tau_{M+1}} \langle c,\xi\rangle ^4 \le 2^{10M^2}$ for $M\ge 1$.
\Dnote{}

The inequality \pref{eq:3} implies the second moment bound $\E_W(\pE W)^2\le 2^{100 kM^2} (\E_W\pE W)^2$.
\Dnote{}
By \pref{lem:second-moment-bound} and \pref{eq:2}, it follows that
\begin{displaymath}
  \Pr_{W} \Set{ \pE W \cdot (1-\langle  c,u \rangle^2) \le O(\e)\pE W} \ge \e^2\cdot 2^{-100 kM^2}=2^{-O(kM^2)}\mper
\end{displaymath}
Since $\{W\}$ has density $2^{-O(M^2)}$ in $\{W_0\}$, it also follows that
\begin{displaymath}
  \Pr_{W_0} \Set{ \pE W_0 \cdot \langle  c,u \rangle^2 \ge (1-O(\e))\pE W_0} \ge 2^{-O(kM^2)}\mper\qedhere
\end{displaymath}

\end{proof}

\begin{lemma}
  \label{lem:second-moment-bound}
  \label{lem:second-moment-truncation}
  Let $\{A,B\}$ be a distribution that satisfies $0\le A\le B$.
  Suppose $\E A \le \e \E B$ and $\E B^2 \le t(\E B)^2$.
  Then, $\Pr\{A\le e^{\delta} \e B \}\ge \delta^2/9t$ for all $0\le \delta\le 1$.
\end{lemma}

\begin{proof}

Let $\Ind_{\mathrm{good}}$ be the 0/1 indicator of the event $\{A\le e^{\delta}\e B\}$ and let $p_{\mathrm{good}}=\E\Ind_{\mathrm{good}}$.
Let $\Ind_{\mathrm{bad}}=1-\Ind_{\mathrm{good}}$ be the 0/1 indicator of the complement.
The expecation of $\Ind_{\mathrm{good}} B$ satisfies the lower bound $\pE \Ind_{\mathrm{good}} B\ge (1-e^{-\delta})\pE B$ because  $\e \E B\ge \E A\ge e^\delta \e \E\Ind_{\mathrm{bad}}  B$ and thus $\pE \Ind_{\mathrm{bad}} B\ge e^{-\delta}\pE B$.
At the same time, we can upper bound the expectation of $\Ind_{\mathrm{good}} B$ in terms of $p_{\mathrm{good}}$ using Cauchy--Schwarz and the second moment bound $\E B^2\le t(\E B)^2$,
\begin{displaymath}
  \E \Ind_{\mathrm{good}} B \le (\E \Ind_{\mathrm{good}}^2 \cdot \E B^2)^{1/2}\le (p_{\mathrm{good}}\cdot t)^{1/2} \E B.
\end{displaymath}
It follows that $p_{\mathrm{good}}\ge (1-e^{-\delta})^2/t\ge \delta^2/9t$.\qedhere

\end{proof}

\begin{lemma}
  Let $c\in \R^n$ be a unit vector and let  $\{u\}$ be a degree-$2$ pseudo-distribution over $\R^n$ that satisfies the constraint $\lVert u \rVert_2^2=1$.
  Suppose $\pE \langle c,u\rangle ^2\ge 1-\e$ for $\e>0$.
  Then, there exists a distribution $\{v\}$ over unit vectors in $\R^n$ such that $\Pr\{\langle  c,v \rangle^2\ge 1-2\e\}=\Omega(1)$.
  Moreover, there exists a randomized polynomial-time algorithm to sample from such a distribution $\{v\}$.
\end{lemma}

\begin{proof}

Let $\{\xi\}$ be the Gaussian distribution with the same first two moments as $\{u\}$ (so that $\E Q(v) = \pE Q(u)$ for every degree-$2$ polynomial $Q$).
(See \pref{lem:moments}.)
\Dnote{}
We choose $v=\xi/\lVert  \xi \rVert_2$.
Since the first two moments of $\{\xi\}$  and $\{u\}$ match, we have $\E (\lVert  \xi \rVert_2^2-\langle c, \xi\rangle^2) \le \e \E \lVert  \xi \rVert_2^2 $.
Since $\{\xi\}$ is a Gaussian distribution, it satisfies $\E\lVert  \xi \rVert_2^4\le O(\E \lVert  \xi \rVert_2^2)^2$.
\Dnote{}
By \pref{lem:second-moment-bound}, it follows that the event $\{\langle  c,\xi \rangle^2 \ge (1-2\e)\lVert  \xi \rVert_2^2\}$ has constant probability.
This event is equivalent to the event $\{ \langle  c, v \rangle^2\ge 1-2\e\}$.
\qedhere

\end{proof}

%% file: tensor-decomposition.tex
\section{Noisy tensor decomposition} \label{sec:tensor-decomp-proof}

In this section we will prove \pref{thm:tensor-decomp} (noisy tensor decomposition, quasi-polynomial time).

\restatetheorem{thm:tensor-decomp}

\noindent
The proof combines the following lemma with \pref{thm:sample} (sampling pseudo-distributions).
The lemma formalizes the following fact in terms of low-degree pseudo-distributions:
the polynomial $\lVert  A^\top u \rVert_d^d\in\R[u]$ assumes large values over the sphere only at points close to one of the columns of $A$.
Note that the conclusion of the lemma allows us to reconstruct a column of $A$ in time $n^{O(k)}$ using \pref{thm:sample} (sampling pseudo-distributions).

\begin{lemma} \label{lem:tensor-decomp-main}
  Let $A$ be a $\sigma$-overcomplete dictionary and let $\{u\}$ be a degree-$3k$ pseudo-distribution over $\R^n$ that satisfies the polynomial constraints $\{\lVert  A^\top u \rVert_d^d\ge e^{-\delta d}, \lVert  u \rVert_2^2=1\}$.
  Then, there exists a column~$c$ of $A$ such that $\pE \langle c,u \rangle^k\ge e^{-\e k}$ for $\e = O(\delta + \tfrac{\log \sigma}{d} + \tfrac{\log m}{k})$.
\end{lemma}

\noindent

\begin{proof}
First, we claim that the pseudo-distribution $\{ u\}$ also satisfies the constraint $\{\lVert  A^\top u \rVert_k^k \ge e^{-\delta' k}\}$ where
$\delta' = \tfrac{d}{d-2}\delta + \frac {\log \sigma} {d-2} $.
The proof of this claim follows by a sum-of-squares version of the following form of \Hoelder's inequality,
\begin{displaymath}
  \lVert  v \rVert_d \le \lVert  v \rVert_k^{1-2/d} \cdot  \lVert  v \rVert_2^{2/d}\mper
\end{displaymath}
(This inequality holds for all norms $\lVert \cdot  \rVert_k$ with $k\ge 1$, including $\lVert  \cdot \rVert_\infty$.)
In particular, if $k$ is an integer multiple of $d-2$, the following relation of degree $k+2k/(d-2)$ holds among polynomials in $\R[v]$,
\begin{displaymath}
  (\lVert  v \rVert_d^d)^{k/(d-2)} \preceq ( \lVert  v \rVert_2^{2})^{k/(d-2)}\cdot \lVert  v \rVert_k^k\mper
\end{displaymath}
See the overview section for a proof of this fact.
\Dnote{}
By substituting $v=A^\top u$ and using the facts that  $\lVert  A^\top u \rVert_2^2 \sleq \sigma \norm{u}^2$ and that $\{u\}$ satisfies the constraint $\{\lVert  u \rVert^2=1\}$, we get that $\{u\}$ satisfies $\{\lVert  A^\top u \rVert_k^k \ge (\lVert  A^\top u \rVert_d^d)^{k/(d-2)}/\sigma^{k/(d-2)}\}$, which implies the claim because $\{\lVert  A^\top u \rVert_d^d\ge e^{-\delta d}\}$.

By an averaging argument, there exists some column $c$ of $A$ that satisfies $\pE \langle  c,u \rangle^k\ge \pE \lVert  A^\top u \rVert_k^k/m \ge e^{-\delta' k}/m=e^{-\e k}$ for $\e=\delta ' + \frac{\log m}{k}$ as desired.
\end{proof}

\paragraph{Proof of \pref{thm:tensor-decomp} from \pref{lem:tensor-decomp-main} and~\pref{thm:sample}}

Our tensor decomposition algorithms constructs a set of unit vectors $S\subseteq \R^m$ in an iterative way.
We will determine the choice of the parameters $k\ge 1$ and $\gamma>0$ later.

\medskip

\begin{compactitem}
\item Start with $S=\emptyset$.
\item While there exists a degree-$k$ pseudo-distribution $\{u\}$ that satisfies the constraints  $\{ P(u) \geq 1-\tau,  \norm{u}_2^2=1 \}$ and $\{ \iprod{s,u}^2 \leq 1-\gamma \}$ for every $s\in S$:
  \begin{compactitem}
  \item Use the algorithm in \pref{thm:sample} (sampling pseudo-distributions) to obtain in time $n^{k/\poly(\e)}$ a unit vector $c'\in \R^m$ that satisfies $P(c')\ge e^{-\e d} -\tau$ for $\e = O(\frac \tau d + \frac{\log \sigma}{d} \frac {\log m}{k} )$ (by \pref{lem:tensor-decomp-main}) and $\langle c',s \rangle^2\le 1-\gamma/10$ for every vector $s\in S$.
  \item Add the vector $c'$ to the set $S$.
  \end{compactitem}
\end{compactitem}

\medskip

\noindent
Let us first explain why we can find a vector $c'$ that satisfies the above conditions if there exists such a pseudo-distribution $\{u\}$.
Recall that the input polynomial $P$ satisfies $\pm (P-\lVert  A^\top u \rVert_d^d) \preceq \tau \lVert  u \rVert_2^d$.
Therefore, the above pseudo-distributions satisfy $\{ \lVert  A^\top u \rVert_d^d \ge 1-2\tau = e^{-\delta d}\}$ for $\delta=O(\tau/d)$.
Hence, \pref{lem:tensor-decomp-main} implies that a column $c$ of $A$ satisfies $\pE \langle  c,u \rangle^k\ge e^{-\e' k}$ for $\e'=O((\frac \tau d + \frac{\log \sigma}{d} \frac {\log m}{k})$.
Thus, the algorithm of \pref{thm:sample} will output a unit vector $c'$ with $\langle  c,c' \rangle^k \ge e^{-O(\e') k}=e^{-\e k}$ with probability at least $2^{-k/\poly(\e)}$.
Note that the condition $\langle  c,c' \rangle^k\ge e^{-\e k}$ implies that $P(c')\ge e^{-\e d }-\tau$.
By repeating the algorithm $2^{k/\poly(\e)}$ times we can ensure that with high probability one of the vectors found in this way satisfies the desired condition.
We claim that the condition $\langle  c,c' \rangle^2 \ge 1-O(\e)$ implies that $\langle  c',s \rangle\le 1-\gamma/10$ for all $s\in S$ (assuming a suitable choice of $\gamma$).
Let $\gamma'=\frac12\lVert  (c')^{\otimes 2} - s^{\otimes 2} \rVert^2$.
We are to show $\gamma'\ge \gamma/10$.
By the triangle inequality, $\lVert c^{\otimes 2}-s^{\otimes 2} \rVert^2\le O(\e) + 4\gamma'$.
Together with an SOS version of the triangle inequalitity,
\begin{math}
  \lVert s^{\otimes 2}-u^{\otimes 2} \rVert^2
  \preceq 8\gamma' + O(\e) + 2 \lVert c^{\otimes 2} -u^{\otimes 2} \rVert^2\mper
\end{math}
Since $\{u\}$ satisfies $\{\langle  s,u \rangle^2 \le 1-\gamma\}$ it follows that $\{2\gamma \le 8\gamma' + O(\e) + 2 \lVert c^{\otimes 2} -u^{\otimes 2} \rVert^2\}$, which implies the constraint $\{\langle  c,u \rangle^2 \le 1-\gamma/2 + 2\gamma' + O(\e)\}$ (using the constraint $\{\lVert  u \rVert^2=1\}$).
However, since $c$ satisfies $\pE \langle  c,u \rangle^k\ge e^{-\e k}$, we have $\gamma/2 -2\gamma' - O(\e) \le  O(\e)$, which means that $\gamma'\ge \gamma/4 -O(\e)\ge \gamma/10$ as desired.
(Here, we are assuming that $\gamma$ was chosen so that $\gamma/\e$ is a large enough constant.)

Next we claim that every vector in $s\in S$ is close to one of the columns of $A$.
Indeed, every such vector satisfies $\lVert  A^\top s \rVert_d^d \ge e^{-\e d }-2\tau $, which by \pref{lem:tensor-decomp-main} implies that $\langle  s,c \rangle^2\ge 1-O(\e + \tau/d + (\log \sigma) /d ) = 1-O(\e)$ for a column $c$ of $A$.

Next we claim that if the algorithm terminates then for every column $c$ of $A$ there exists a vector $s\in S$ with $\langle  c,s \rangle^2 \ge 1-\gamma$.
Indeed, if there exists a column that violates this condition, then it would satisfy all constraints for the pseudo-distribution, which means that the algorithm does not terminate at this point.

To finish the proof of the theorem it remains to bound the number of iterations of the algorithm.
We claim that the number of iterations is bounded by the number $m$ of columns of $A$ because in each iteration the vectors in $S$ will cover at least one more of the columns of $A$.
As observed before, every vector $s\in S$ is close to a column~$c_s$ of $A$ in the sense that $\lVert  s^{\otimes 2}-c_s^{\otimes 2} \rVert^2 = O(\e)$.
However, since $c'$ satisfies $\langle  c',s \rangle^2 \le 1-\gamma/10$, we have by triangle inequality $\gamma/5 \le \lVert  (c')^{\otimes 2} -s^{\otimes 2} \rVert^2 \le 2 \lVert  (c')^{\otimes 2} - c_s^{\otimes 2} \rVert^2  + 2 \lVert  s^{\otimes 2} - c_s^{\otimes 2} \rVert^2$, which means that $\lVert  (c')^{\otimes 2}-c_s^{\otimes 2} \rVert^2\ge \gamma/10 - O(\e)$.
Therefore, the vector $c'$ is not close to any of the vectors $c_s$ for $s\in S$, which means that it has to be close to another column of $A$.
(Here, we are again assuming that $\gamma$ was chosen so that $\gamma/\e$ is a large enough constant.)
\qed

\begin{remark}[Handling columns with varying norms] \label{rem:different-norms}
Many of our techniques also apply to dictionaries with columns of different $\ell_2$ norms.
In particular, using the same algorithm, we can reconstruct  in this case a single vector close to one of the columns.
More generally, we can reconstruct a set of vectors that is close to the set of columns with maximum norm.

By adapting the algorithm somewhat we can also achieve recovery guarantees for columns with significantly smaller norm than the maximum norm.
Concretely, we can modify the algorithm so that we ask for pseudo-distributions satisfying $P(u) \geq \rho$, where $\rho$ is a parameter that we gradually decrease so we can get all the vectors.
However, we need to also change the right-hand side of the constraint $\iprod{u,s}^2 \leq 1-\gamma $ to a value that decreases with $\rho$.
Otherwise, the algorithm might not terminate, as there can be exponentially vectors that are somewhat far from a column vector $c$, and all of them will have fairly large value for $P(\cdot)$.
Such a modified algorithm can still obtain all the column vectors (up to a small error) if we assume that the they are sufficiently \emph{incoherent}.
That is, $\iprod{a,a'} \leq \mu$ for every distinct columns $a$, $a'$ of $A$ with $\mu$ depending on the norm ratios.
Similar (and in fact often stronger) assumptions were made in prior works on dictionary learning.
(However, we need these assumptions only when the vectors have different norms.)
\end{remark}

%% file: polytime.tex
\section{Polynomial-time algorithms}
\label{sec:polytime}

In this section we show how we can improve our tensor decomposition algorithm when we have access to examples of very sparse linear combinations of the dictionary columns,
culminating in Theorem~\ref{thm:dict-learn-poly} that gives a polynomial-time algorithm for the dictionary problem for the case the distribution is $(d,\tau)$-nice for $\tau = n^{-\Omega(1)}$.

\subsection{Sampling pseudo-distributions}

The following theorem refines \pref{thm:sample} (sampling pseudo-distributions) reconstructing a vector $c'$ that is close to a target vector $c$.
We make an additional assumption about having access to samples from a distribution $\{W\}$ over sum-of-squares polynomials.
This distribution comes with a noise parameter $\tau$ that controls how well the distribution correlated with the target vector $c$.
If this noise parameter is sufficiently small, samples from distribution allow the algorithm to work under a more refined but milder condition on the pseudo-distribution $\{u\}$.
For our dictionary learning algorithm, we can satisfy this condition when the noise parameter $\tau$ of the distribution $\{W\}$ satisfies $\tau\ll m^{1/k}$.
(The noise parameter $\tau$ roughly coincides with the niceness parameter of the distribution $\{x\}$.)

\Dnote{}

\begin{theorem}[refined sampling from pseudo-distributions]
\label{thm:refined-reweighing}
  For every $k\ge 1$, there exists a $n^{O(k)}$-time algorithm with the following guarantees:
  Suppose the input of the algorithm is a pseudo-distribution  $\{u\}$ over $\R^n$ and a sum-of-squares polynomial $W\in\R[u]$ satisfying the following properties for some unit vector $c\in \R^n$:
  \begin{itemize}
  \item The sum-of-squares polynomial $W$ is chosen from a distribution $\{W\}$ with with mean $\bar W = \E_W W$  and second moment $\E_W W(u)W(u')\preceq M\cdot \bar W(u) \cdot \bar W(u')$ satisfying
    \begin{equation}
      \label{eq:nice-reweighing}
      \langle  c,u \rangle^{2(1+k)} \preceq \bar W \preceq \cramped{(\langle  c,u \rangle^2 + \tau \lVert  u \rVert_2^2)^{1+k}}\mper
    \end{equation}
  \item The pseudo-distribution $\{u\}$ has degree $2(1+2k)$ and satisfies the polynomial constraint $\lVert  u \rVert_2^2=1$ and the conditions
    \begin{equation}
      \label{eq:refined-condition}
      \pE \langle  c,u \rangle^{2(1+2k)}\ge e^{-\e k}\pE \langle  c,u \rangle^{2}\text{ and }\pE \langle  c,u \rangle^{2}\ge \tau^{k}.
    \end{equation}
  \end{itemize}
  Then, the output of the algorithm is a unit vector $c'\in\R^n$ such that with probability at least $\tau^2/M 2^{O(k)/\poly(\e)}$,
  \begin{displaymath}
    \langle  c,c' \rangle^2\ge e^{-O(\e+3(1/k+1)\tau)}\mper
  \end{displaymath}
\end{theorem}

The following lemma is the main new ingredient of the proof of this theorem.

\begin{lemma}
\label{lem:refined-reweigh}
  Let $\{u\}$ be a degree-$2(1+2k)$ pseudodistribution that satisfies the constraint $\lVert  u \rVert_2^2=1$.
  Let $\{W\}$ be a distribution over sum-of-squares polynomials.
  Suppose $\{u\}$ and $\{W\}$ satisfy the conditions in \pref{thm:refined-reweighing}.
  Then, $\pE_u W \cdot \langle c,u\rangle^{2k} \ge  e^{-\e' k}\pE_u W$ with probability $\tau^2/M2^{O(k)}$ over the choice of $W$ for $\e'=\e+3(1/k+1)\tau$.
\end{lemma}

Note that the conclusion of the lemma implies that we can recover a vector $c'$ with $\langle  c',c \rangle^2\ge 1-O(\e')$ using \pref{thm:sample} in time $n^{k /poly(\e')}$ with probability $2^{O(k)/\poly(\e')}$.
Therefore, \pref{thm:refined-reweighing} follows by combining \pref{lem:refined-reweigh} with \pref{thm:sample}.

\begin{proof}[Proof of \pref{lem:refined-reweigh}]
We will show that the polynomials $\bar W\langle c,u\rangle ^k$ and $\bar W$ have similar pseudo-expectations by comparing them to the polynomials $\langle c,u\rangle ^2$.
We will show that $\pE \bar W \langle  c,u \rangle^k$
For brevity, choose polynomials $\alpha=\langle c,u\rangle^2\in\R[u]$ and $\beta=\lVert  u \rVert^2\in\R[u]$ so that $0\preceq \alpha\preceq \beta$.
Then,
\begin{multline}
  \alpha^{1+k}
  \preceq \bar W
  \preceq (\alpha+\tau \beta)^{1+k}
  =\alpha\sum_{i=0}^{k} \binom{1+k}{i} \alpha^{k-i} (\tau \beta)^{i} + (\tau\beta)^{k+1}
  \\
  \preceq \alpha \beta^k \sum_{i=0}^k (1+k)^i\tau^{i} + \tau^{k+1} \beta^{k+1}
  \preceq \bigl(1+2(t+k)\tau\bigr)\alpha \beta^k + \tau^{k+1}\beta^{k+1}\mper
\end{multline}
Here, the last step uses the assumption $(1+k)\tau\le 1/2$ to bound the series $\sum_{i=1}^k (1+k)^i\tau^i\le 2 (t+k)\tau$.
It follows that
\begin{displaymath}
  \pE \bar W \langle c,u\rangle^k\ge \pE \alpha^{1+2k}\ge e^{-\e k}\pE \alpha\mper
\end{displaymath}
(Here, we used \pref{eq:refined-condition})
At the same time,
\begin{displaymath}
  \pE \bar W
  \le \bigl(1+2(1+k)\tau\bigr)\pE \alpha + \tau^{k+1}
  \le \bigl(1+2(1+k+1)\tau\bigr)\pE \alpha
  \le e^{2(2+k)\tau} \pE \alpha\mper
\end{displaymath}
Here, the second step uses the assumption $\tau^{k+1}\le \tau\pE \langle c,u \rangle^{2}$.
Together, the two bounds imply $\pE \bar W \langle c,u \rangle^{2k}\ge \cramped{e^{- \e k -2(2+k)\tau}}\pE \bar W$.

In order to lower bound the probability of the event $\{ \pE W \langle c,u  \rangle^{2k}\ge e^{-\e'k}\pE W\}$, we will upper bound the second moment $\E (\pE W)^2$ and apply \pref{lem:second-moment-bound}.
By the premise $\E W(u) W(u') \preceq M\cdot \bar W(u)\bar W(u')$, we get
\begin{displaymath}
  \E \bigl(\pE W\bigr)^2 = \pE_{\{u\}\{u'\}} \E W(u) W(u') \preceq \pE _{\{u\}\{u'\}} M\cdot \bar W(u)\bar W(u')=M\cdot \bigl(\pE \bar W\bigr)^2\mper
\end{displaymath}
By \pref{lem:second-moment-bound}, the probabilility of the event $\{ \pE W \langle  c,u \rangle^{2k} \ge (e^{-\e k - 2(2+k)\tau }-\delta)\pE W \}$ is at least $\Omega(\delta^2/M)$.
We choose $\delta=\tau 2^{-O(k)}$ to lower bound the probability of the event $\{ \pE W \langle  c,u \rangle^{2k} \ge e^{-\e' k}\pE W \}$ for $\e'=\e+3(1/k+1)\tau$ by $\Omega(\tau^2/M 2^{O(k)})$.
\end{proof}

\subsection{Tensor decomposition}

The following lemma shows that a pseudo-distribution $\{u\}$ that satisfies the constraints $\{\lVert  A^\top u \rVert_{2(t+k)}^{2(t+k)} \approx 1,  \lVert u \rVert_2^2=1\}$ also satisfies the condition of \pref{thm:refined-reweighing} for one of the columns of the dictionary $A$.

\begin{lemma}
\label{lem:tensor-poly-time}
    Let $A\in\R^{n\times m}$ be a $\sigma$-overcomplete dictionary and let $\{u\}$ be a degree-$2(k+t)$ pseudo-distribution over $\R^n$ that satisfies $\lVert  u \rVert_2^2=1$.
  Suppose $\{u\}$ also satisfies the polynomial constraint $\lVert  A^\top u \rVert_{2(1+k)}^{2(1+k)}\ge e^{-2(k-1)\e}\sigma$ for $k\ge 1$.
  Then, there exists a column~$c$ of $A$ such that $\pE \langle c,u \rangle^{2k}\ge e^{-2 \e k}\pE \langle c,u \rangle^{2}$ and $\pE\langle c,u \rangle^{2}\ge \e e^{-2k\e}/m$.
\end{lemma}

\noindent
\emph{Remark.}
For the lower bound on $\pE\langle c,u \rangle^{2}$, we typically only need that it is polynomial.
The algorithm in \pref{thm:refined-reweighing} allows us to recover a vector close to $c$ in time $n^t$ assuming that $\tau^k\ll 1/m$.

\begin{proof}

We will prove the contrapositive.
Let $a_1,\ldots,a_m \in\R^n$  be the columns of $A$ and let $\e'=\e e^{-2k\e}$.
Suppose every column~$c$ satisfies  either  $\pE \langle c,u \rangle^{2(1+k)}< e^{-2\e k}\pE \langle c,u \rangle^{2}$ or $\pE\langle c,u \rangle^{2}< \e'/m$.
We are to show that the pseudo-distribution $\{u\}$ cannot satisfy the constraint $\lVert  A^\top u \rVert_{2(1+k)}^{2(1+k)}\ge e^{2(k-1)\e}\sigma$.
Indeed, these conditions allow us to upper bound
\begin{displaymath}
\pE \lVert  A^\top u \rVert_{2(1+k)}^{2(1+k)}
=\sum_{i=1}^m \pE \langle  a_i, u \rangle^{2(1+k)}
\le e^{-2\e k}\pE \lVert  A^\top u \rVert_2^2 + \e'
\le  (1+\e)e^{-2k\e}\sigma.
\end{displaymath}
It follows that the pseudo-distribution $\{u\}$ cannot satisfy the constraint $\lVert  A^\top u \rVert_{2(1+k)}^{2(1+k)}\ge e^{2(k-1)\e}\sigma$.

\end{proof}

\subsection{Dictionary learning}

The following lemma shows that up to polynomial reweighing the distribution $\{y=A x\}$ gives us access to a distribution $\{W\}$ that satisfies the condition of \pref{thm:refined-reweighing}.

\begin{lemma}
\label{lem:reweighing}
  Let $A\in\R^{n\times m}$ be a $\sigma$-overcomplete dictionary and  and let $\{x\}$ be a $(k,\tau)$-nice distribution over $\R^m$ with $k\ge 4$.
  For $i\in [m]$, let $\cD_i$ be the distribution obtained from reweighing the distribution $\{w=c\langle  Ax ,u \rangle^2\}$ by $x_i^2$, where $c=\E x_i^2/\E x_i^4$.
  Then, $ \langle  \super a i, u \rangle^2 \preceq \E_{\cD_i} w \preceq \langle  \super a i, u \rangle^2 + \tau \sigma \lVert  u \rVert_2^2$.
\end{lemma}

\begin{proof}
The expectation of $w$ after reweighing by $x_i^2$ satisfies
\begin{displaymath}
  \E_{\cD_i} w
  = \tfrac 1 {\E_{\{x\}} x_i^2}\E_{\{x\}} x_i^2 \cdot c \langle  A x, u \rangle^2
  =  \sum_{j} \tfrac 1{\E_{\{x\}} x_i^4} \E_{\{x\}} x_i^2x_j^2 \langle  \super a i,u \rangle^2 \langle  \super a j,u \rangle^2
\end{displaymath}
The last step uses that all non-square moments of $\{x\}$ vanish.
The desired bounds follow because the coefficient of $\langle  \super a i,u \rangle^2$ is $1$ and  for all indices $j\neq i$, the coefficients of $\langle  \super a j,u \rangle^2$ are all between $0$ and $\tau$.
For the final bounds, we also use $\lVert  A^\top u \rVert_2^2\preceq \sigma \lVert  u \rVert^2$.
\end{proof}

\begin{theorem}[Dictionary learning, polynomial time, single dictionary vector]
\label{thm:dict-poly-single}
  There exists an algorithm that solves the following problem for every desired accuracy $\e>0$, overcompletess $\sigma \ge 2$, in time $n^{O(k)}$ with success probability $n^{-O(k)/\poly(\e)}$ for noise $\tau\le O(\e)$, where $k= (1/\e) \log \sigma + \tfrac{\log m}{\log (1/\tau)}$:
  Given $k$ samples from a distribution of the form $\{y=Ax\}$ and a degree-$k$ pseudo-distribution $\{u\}$ that satisfies $\{\lVert  A u \rVert_{k}^k\ge e^{-\e k},\lVert  u \rVert_2^2=1\}$, where $A$ is a $\sigma$-overcomplete dicionary and $\{x\}$ is a $(4,\tau)$-nice distribution, output a unit vector $c'$ such that there exists a column $c$ of $A$ with $\langle c,c'\rangle ^2\ge 1-O(\e)$ and $\pE \langle  c,u \rangle^{k}\ge e^{-O(\e)k}\pE \langle  c,u \rangle^{2}$.
\end{theorem}

\begin{proof}

We run the algorithm in \pref{thm:refined-reweighing} on the pseudo-distribution $\{u\}$ and the following distribution $\{W\}$ over squared polynomials:
Choose $k'=k/2-1$ independent samples $y_1,\ldots,y_{k'}$ from the distribution $\{y=Ax\}$ and form the degree-$(k-2)$ polynomial $W=\langle  y_1,u \rangle^2\cdots\langle  y_{k'},u \rangle^2$.
This distribution $\{W\}$ does not satisfy the condition in \pref{thm:refined-reweighing} but it turns out to be sufficiently close to a distribution that satisfies the condition.
Let us first verify that the pseudo-distribution $\{u\}$ satisfies the condition of \pref{thm:refined-reweighing} for a vector $c$ as in the theorem above.
Indeed, by \pref{lem:tensor-poly-time}, there exists a column $c$ of $A$ such that $\pE \langle  c,u \rangle^k\ge e^{-O(\e) k}\pE \langle  c,u \rangle^2$ and $\pE \langle  c,u \rangle^2\ge O(\e) e^{-O(\e k)} / m\ge \tau^k$.
(Since $k\ge (1/\e)\log \sigma$, the pseudo-distribution $\{u\}$ satisfies the constraint $\{\lVert  A^\top u \rVert_k^k\ge e^{-O(\e) k}\sigma\}$ as required by \pref{lem:tensor-poly-time}.)
It follows that if we run the algorithm in \pref{thm:refined-reweighing} for a distribution over polynomials that satisfies condition \pref{eq:nice-reweighing} for this column $c$ of the dictionary $A$, then the algorithm outputs a vector $c'$ with the above properties with significant probability.

We will use \pref{lem:reweighing} to reason about the distribution $\{W\}$.
Without loss of generality, we assume that $c$ is the first column of the dictionary $A$.
Let $\bar x=(x_1,\ldots,x_{k'})$ be $k'$ independent samples from $\{x\}$.
(The distribution $\{W\}$ is the same as $\{ \langle A x_1, u \rangle^2\cdots \langle  A x_{k'},u \rangle^2\}$.)
We claim that the distribution $\{W\}$ satisfies \pref{eq:nice-reweighing} after reweighing by the function $r(\bar x)^2=x_{1,1}^2\cdots x_{k',1}^2$ (the product of the square of the first coordinates of $x_1,\ldots,x_{k'}$).
The distribution after reweighing is, up to scaling of the polynomials, equal to the distribution $\cD=\{W=w_1\cdots w_{k'}\}$, where $w_1,\ldots,w_{k'}$ are independent samples from the distribution $\cD_1$ in \pref{lem:reweighing}.
By \pref{lem:reweighing}, this reweighted distribution satisfies the condition \pref{eq:nice-reweighing}, that is,
\begin{displaymath}
  \langle c,u \rangle^{2k'}  \preceq \E_{\cD} W \preceq (\langle  c,u \rangle^2+\tau \sigma \lVert  u \rVert_2^2)^{k'}\mper
\end{displaymath}
Since we assume $\{x\}$ to be $(4,\tau)$-nice, the variance of $\cD$ is bounded by $n^{O(k)}$.
\Dnote{}

Let $\cA$ be the algorithm in \pref{thm:refined-reweighing}.
Since $\cD$ satisfies the conditions of \pref{thm:refined-reweighing}, if we run $\cA$ on the pseudo-distribution $\{u\}$ and the distribution $\cD$ over polynomials, it will succeed with probability $n^{-O(k)/\poly(\e)}$.
We claim that the success probability on the distribution $\{W\}$ (before reweighing) is comparable.
Let $p(W)$ be the probability that the algorithm succeeds for a particular input polynomial $W$.
Under the distribution $\cD$, algorithm $\cA$ has success probability $\E_{\cD}p(W)\ge n^{-O(k)/\poly(\e)}$.
We relate this success probability to the success probability under $\{W\}$ as follows,
\begin{displaymath}
  n^{-O(k)/\poly(\e)}\le \E_{\cD} p(W)
  = \tfrac 1{\E_{\{\bar x\}} r(x)^2 } \E_{\{\bar x\}} r(\bar x)^2 p(W)
  \le   \tfrac 1{\E_{\{\bar x\}} r(x)^2 } \Paren{\E_{\{\bar x\}} r(\bar x)^4 p(W) \cdot \E_{\{W\}} p(W)}^{1/2}\mcom
\end{displaymath}
where the last step uses Cauchy--Schwarz.
The niceness property of $\{x\}$ implies that $\E_{\bar x} r(\bar x)^4 p(W)\le \E_{\bar x} r(\bar x)^4=  (\E_{\{x\}} x_1^4)^{k'}=n^{O(k)}\cdot (\E_{\{x\}}x_1^2)^{2k'}=n^{O(k)}(\E_{\{\bar x\}} r(\bar x)^2)^2$.
Therefore,  the success probability of $\cA$ under the distribution $\{W\}$ (before reweighing)
satisfies $\E_{\{W\}} p(W)\ge n^{-O(k)/\poly(\e)}$.

\end{proof}

\Dnote{}
\Bnote{}

The following theorem gives a polynomial time algorithm for dicionary learning under $(d,\tau)$-nice distributions for all $\tau=n^{\Omega(1)}$.

\begin{theorem}[Dictionary learning, polynomial time]
\label{thm:dict-learn-poly}
There exists an algorithm that for every desired accuracy $\e>0$ and overcompleteness $\sigma\ge 1$ solves the following problem for every $(d,\tau)$-nice distribution with $d \ge d(\e,\sigma)=O(d^{-1}\log \sigma )$ and $\tau\le \tau(\e,\sigma)=(\e^{-1}\log \sigma )^{O(\e^{-1}\log \sigma )}$ in time $n^{(1/\e)^{O(1)} k}$ for $k=d + O(\tfrac{\log m} {\log (1/\tau)})$:
Given $n^{O(d)}/\poly(\tau)$ samples from a distribution $\{y=Ax\}$ for a $\sigma$-overcomplete dictionary $A$ and $(d,\tau)$-nice distribution $\{x\}$, output a set of vectors that is $\e$-close to the set of columns of $A$ (in symmetrized Hausdorff distance).

\end{theorem}

\begin{proof}

We will show how to use \pref{thm:dict-poly-single} to recover a single vector that is close to one of the columns of $A$.
By repeating this step in the same way as in the proof of \pref{thm:tensor-decomp} (noisy tensor decomposition) we can recover a set of vectors that is close to the set of columns of $A$.

To recover a single vector, we estimate from the samples of $\{y=Ax\}$ a polynomial $P$ that is close to $\lVert  A^\top  u \rVert_d^d$ in the same way as in the proof of \pref{thm:dict-learn}.
(The distance of $P$ from $\lVert  A^\top u \rVert_d^d$ in spectral norm will be $O(\tau d^d)=O(\e)$.)
Next, we compute a degree-$k$ pseudo-distribution $\{u\}$ that satisfies the constraints $\{P\ge 1-\e, \lVert  u \rVert_2^2=1\}$.\footnote{To recover all vectors, we would also add constraints $\{\langle  s,u \rangle^2\le 1-\gamma\}$ for all vectors $s$ that have already been recovered (see proof of \pref{thm:tensor-decomp}).}
The same argument as in the proof of \pref{lem:tensor-decomp-main} shows that $\{u\}$ also satisfies the constraint $\{\lVert  A u \rVert_k^k\ge e^{O(\e)k}\}$, which means that $\{u\}$ satisfies the premise of \pref{thm:dict-poly-single}.
Therefore, the algorithm in \pref{thm:dict-poly-single} recovers a vector close to one of the columns of $A$.

\end{proof}

\Dnote{}

%% file: conclusion.tex
\section{Conclusions and Open Problems}
\label{sec:conclusion}

The \emph{Sum of Squares} method has found many uses across a variety of disciplines, and in this work we demonstrate its potential 
for solving unsupervised learning problems in regimes that have so far eluded other algorithms. 
It is an interesting direction to identify other problems that can be solved using this algorithm.

The generality of the SOS method comes at a steep cost of efficiency. 
It is a fascinating open problem, and one we are quite optimistic about, to use the ideas from the SOS-based algorithm
to design practically efficient algorithms.

%% file: monomial-inequality.tex
\newcommand{\RC}{\ensuremath{\mathcal{R}}}
\newcommand{\q}{\ensuremath{q}}
\newcommand{\AMGM}{\ensuremath{Q}}
\newcommand{\parts}[1]{\ensuremath{P_{#1}}}
\newcommand{\sym}{R}

\section{Proof of Lemma~\ref{lem:monomial-ineq}}
\label{sec:monomial-ineq}
Lemma~\ref{lem:monomial-ineq} is a consequence of the following sum-of-squares version of the AM-GM inequality.%
\footnote{The first sum-of-squares proof of the AM-GM inequality dates back to Hurwitz in 1891~\cite{Hurwitz1891}.  For related results and sums-of-squares proofs of more general sets of inequalities, see \cite{ Reznick87, Reznick89,frenkel14}.}

\Dnote{}
\begin{lemma}\label{lem:AMGM}
Let $w_1,\dots,w_n$ be polynomials.  Suppose $w_1,\ldots,w_n\succeq 0$. Then,
$$\frac{w_1^n+\dots +w_n^n}{n} \succeq w_1 w_2\cdots w_n\mper$$
\end{lemma}

To see that this lemma implies \pref{lem:monomial-ineq}, write for a multi-index $\alpha$ with $\card{\alpha}=s$ the polynomial
$w^\alpha$ as a product %
$w^\alpha=\prod_{j=1}^{s} w_{i_j},$
where $w_i$ is repeated $\alpha_i$ times.
(E.g., we would write $w_1^2 w_2 w_3^2$  as $w_1 w_1 w_2 w_3 w_3$ and we would have $(i_1,\dots,i_5)=(1,1,2,3,3)$.)
Then applying \pref{lem:AMGM} to the polynomials $w_{i_1},\ldots,w_{i_s}$ gives the inequality asserted in \pref{lem:monomial-ineq},
$$w^\alpha = w_{i_1}\cdots w_{i_s} \preceq  \frac{w_{i_1}^s+\dots+w_{i_s}^{s}}{s} = \sum_i \frac {\alpha_i}{\card{\alpha}} w_i \sleq \sum_i w_i^s,$$
where the second inequality uses that $0\le \alpha_i/\card{\alpha}\le 1$ and the premise $w_i\sgeq 0$.

\medskip

\paragraph{Proof of \pref{lem:AMGM}}

To prove \pref{lem:AMGM}, we will give a sequence of polynomials
$\sym_0,\dots,\sym_{n-1}$ such that
$\sym_0 =(z_1^n+\dots z_n^n)/n$, $\sym_{n-1}=z_1\dots z_n$, and
 $\sym_0 \sgeq \dots \sgeq \sym_{n-1}$.
To this end, let
$$R_k=\frac{1}{n!}\sum_{\sigma\in S_n} w_{\sigma_1}^{n-k} \prod_{j=2}^{k+1} w_{\sigma_j}
,$$
where $S_n$ denotes the symmetric group on $n$ elements.
So, for instance,
\begin{align*}
R_0&=\frac{1}{n!}\sum_{\sigma\in S_n} w_{\sigma_1}^{n}
    =\frac{1}{n}\left(w_1^n+\dots+w_n^n\right),\\
R_1&=\frac{1}{n!}\sum_{\sigma\in S_n} w_{\sigma_1}^{n-1} w_{\sigma_2}
    =\frac{1}{n(n-1)}\left(w_1^{n-1}w_2+w_1^{n-1}w_3+w_1^{n-1}w_4+\dots+w_n^{n-1} w_{n-1}\right),\\
R_2&=\frac{1}{n!}\sum_{\sigma\in S_n} w_{\sigma_1}^{n-2} w_{\sigma_2} w_{\sigma_3}
    =\frac{1}{n\,\binom{n-1}{2}}\left(w_1^{n-2}w_2 w_3+w_1^{n-2}w_2 w_4 +\dots+w_n^{n-2} w_{n-2}w_{n-1}\right),
 \text{\ \ and }\\
R_{n-1}&=\frac{1}{n!}\sum_{\sigma\in S_n} w_{\sigma_1} w_{\sigma_2}\cdots w_{\sigma_n}=w_1 w_2\cdots w_n.
\end{align*}
The following claim will then complete the proof:
\begin{claim}
For any $k\in\{1,\dots,n-1\}$,
$R_{k-1}-R_k$ is a sum of squares.
\end{claim}
\begin{proof}
For a given permutation $\sigma\in S_n$, the corresponding monomials in $R_{k}$ and $R_{k-1}$ will share many of the same variables, differing only in the exponents of $w_{\sigma_1}$ and $w_{\sigma_{k+1}}$.
We will thus try to arrange the terms of $R_{k-1}-R_{k}$ so that we can pull out the common variables, which will let us reduce our inequality to one involving only two variables.
\begin{align*}
R_{k-1}-R_k
&=\frac{1}{n!}\sum_{\sigma\in S_n}
    \left(
        \Bigg(w_{\sigma_1}^{n-k+1} \prod_{j=2}^{k} w_{\sigma_j}\Bigg)-
        \Bigg(w_{\sigma_1}^{n-k} \prod_{j=2}^{k+1} w_{\sigma_j}\Bigg)
    \right)\\
&=\frac{1}{n!}\sum_{\sigma\in S_n}
    w_{\sigma_1}^{n-k}
    \left(
         w_{\sigma_1}-w_{\sigma_{k+1}}
    \right)
    \Bigg(\prod_{j=2}^{k} w_{\sigma_j}\Bigg)\\
&=\frac{1}{n!}
\sum_{\substack{a,b\in [n]\\ a\neq b}}
\sum_{\substack{\sigma\in S_n\\ \sigma_1=a\\ \sigma_{k+1}=b}}
    w_{\sigma_1}^{n-k}
    \left(
         w_{\sigma_1}-w_{\sigma_{k+1}}
    \right)
    \Bigg(\prod_{j=2}^{k} w_{\sigma_j}\Bigg)\\
&=\frac{1}{n!}
\sum_{\substack{a,b\in [n]\\ a\neq b}}
    w_{a}^{n-k}
    \left(
         w_{a}-w_{b}
    \right)
\sum_{\substack{\sigma\in S_n\\ \sigma_1=a\\ \sigma_{k+1}=b}}
 \prod_{j=2}^{k} w_{\sigma_j}\\
 &=\frac{1}{n!}
\sum_{\substack{a,b\in [n]\\ a< b}}
    \big(w_{a}^{n-k}-w_b^{n-k}\big)
    \left(
         w_{a}-w_{b}
    \right)\cdot
\left\{\rule{0cm}{1.1cm}\right.
\sum_{\substack{\sigma\in S_n\\ \sigma_1=a\\ \sigma_{k+1}=b}}
 \prod_{j=2}^{k} w_{\sigma_j}
 \left.\rule{0cm}{1.1cm}\right\}.
\end{align*}
Since the $w_{i}$ are sums of squares, the expression inside the braces is as well.
It is therefore enough to show that
    $\left(w_{a}^{n-k}-w_b^{n-k}\right)
    \left(
         w_{a}-w_{b}
    \right)$
is a sum of squares.
This follows from the fact that
$$w_{a}^{n-k}-w_b^{n-k}=(w_a-w_b)\left(w_a^{n-k-1}+w_a^{n-k-2}w_b +\dots+w_b^{n-k-2}w_a+w_b^{n-k-1}\right),$$
and thus
 $$\big(w_{a}^{n-k}-w_b^{n-k}\big)
    \left(
         w_{a}-w_{b}
    \right)
    =
    (w_a-w_b)^2\left(w_a^{n-k-1}+w_a^{n-k-2}w_b +\dots+w_b^{n-k-2}w_a+w_b^{n-k-1}\right).$$
\end{proof}

%% file: soslearning.bbl
\newcommand{\etalchar}[1]{$^{#1}$}
\def\cprime{$'$} \def\cprime{$'$} \def\cprime{$'$} \def\cprime{$'$}
  \def\cprime{$'$} \def\cprime{$'$} \def\cprime{$'$} \def\cprime{$'$}
  \def\cprime{$'$} \def\cprime{$'$} \def\cprime{$'$}
  \def\polhk#1{\setbox0=\hbox{#1}{\ooalign{\hidewidth
  \lower1.5ex\hbox{`}\hidewidth\crcr\unhbox0}}} \def\cprime{$'$}
  \def\cprime{$'$} \def\cprime{$'$} \def\cprime{$'$} \def\cprime{$'$}
  \def\cprime{$'$} \def\cprime{$'$} \def\cprime{$'$} \def\cprime{$'$}
  \def\cprime{$'$} \def\cprime{$'$}
  \def\cfac#1{\ifmmode\setbox7\hbox{$\accent"5E#1$}\else
  \setbox7\hbox{\accent"5E#1}\penalty 10000\relax\fi\raise 1\ht7
  \hbox{\lower1.15ex\hbox to 1\wd7{\hss\accent"13\hss}}\penalty 10000
  \hskip-1\wd7\penalty 10000\box7} \def\cprime{$'$} \def\cprime{$'$}
  \def\cprime{$'$} \def\cprime{$'$} \def\cprime{$'$} \def\cprime{$'$}
  \def\ocirc#1{\ifmmode\setbox0=\hbox{$#1$}\dimen0=\ht0 \advance\dimen0
  by1pt\rlap{\hbox to\wd0{\hss\raise\dimen0
  \hbox{\hskip.2em$\scriptscriptstyle\circ$}\hss}}#1\else {\accent"17 #1}\fi}
\providecommand{\bysame}{\leavevmode\hbox to3em{\hrulefill}\thinspace}
\providecommand{\MR}{\relax\ifhmode\unskip\space\fi MR }
% \MRhref is called by the amsart/book/proc definition of \MR.
\providecommand{\MRhref}[2]{%
  \href{http://www.ams.org/mathscinet-getitem?mr=#1}{#2}
}
\providecommand{\href}[2]{#2}
\begin{thebibliography}{YWHM08}

\bibitem[AAJ{\etalchar{+}}13]{AgarwalA0NT13}
Alekh Agarwal, Animashree Anandkumar, Prateek Jain, Praneeth Netrapalli, and
  Rashish Tandon, \emph{Learning sparsely used overcomplete dictionaries via
  alternating minimization}, CoRR \textbf{abs/1310.7991} (2013).

\bibitem[AAN13]{AgarwalAN13}
Alekh Agarwal, Animashree Anandkumar, and Praneeth Netrapalli, \emph{Exact
  recovery of sparsely used overcomplete dictionaries}, CoRR
  \textbf{abs/1309.1952} (2013).

\bibitem[ABGM14]{AroraBGM14}
Sanjeev Arora, Aditya Bhaskara, Rong Ge, and Tengyu Ma, \emph{More algorithms
  for provable dictionary learning}, CoRR \textbf{abs/1401.0579} (2014).

\bibitem[AFH{\etalchar{+}}12]{AnandkumarFHKL12}
Anima Anandkumar, Dean~P. Foster, Daniel Hsu, Sham Kakade, and Yi-Kai Liu,
  \emph{A spectral algorithm for latent dirichlet allocation}, NIPS, 2012,
  pp.~926--934.

\bibitem[AGM12]{AroraGM12}
Sanjeev Arora, Rong Ge, and Ankur Moitra, \emph{Learning topic models - going
  beyond svd}, FOCS, IEEE Computer Society, 2012, pp.~1--10.

\bibitem[AGM13]{AroraGM13}
\bysame, \emph{New algorithms for learning incoherent and overcomplete
  dictionaries}, arXiv preprint 1308.6723 (2013),
  \url{http://arxiv.org/abs/1308.6273}.

\bibitem[Bar98]{Barthe98}
Franck Barthe, \emph{On a reverse form of the brascamp-lieb inequality},
  Inventiones mathematicae \textbf{134} (1998), no.~2, 335--361.

\bibitem[BBH{\etalchar{+}}12]{BarakBHKSZ12}
Boaz Barak, Fernando G. S.~L. Brand{\~a}o, Aram~Wettroth Harrow, Jonathan~A.
  Kelner, David Steurer, and Yuan Zhou, \emph{Hypercontractivity,
  sum-of-squares proofs, and their applications}, STOC, 2012, pp.~307--326.

\bibitem[BCMV14]{BhaskaraCMV14}
Aditya Bhaskara, Moses Charikar, Ankur Moitra, and Aravindan Vijayaraghavan,
  \emph{Smoothed analysis of tensor decompositions}, STOC, 2014.

\bibitem[BCV14]{BhaskaraCV14}
Aditya Bhaskara, Moses Charikar, and Aravindan Vijayaraghavan, \emph{Uniqueness
  of tensor decompositions with applications to polynomial identifiability},
  COLT (Maria-Florina Balcan and Csaba Szepesv{\'a}ri, eds.), JMLR Proceedings,
  vol.~35, JMLR.org, 2014, pp.~742--778.

\bibitem[BKS14]{BarakKS14}
B.~Barak, J.A. Kelner, and D.~Steurer, \emph{Rounding sum of squares
  relaxations}, STOC, 2014.

\bibitem[BS14]{BarakS14}
Boaz Barak and David Steurer, \emph{Sum-of-squares proofs and the quest toward
  optimal algorithms}, Proceedings of International Congress of Mathematicians
  (ICM), 2014, To appear.

\bibitem[Com94]{comon1994independent}
Pierre Comon, \emph{Independent component analysis, a new concept?}, Signal
  processing \textbf{36} (1994), no.~3, 287--314.

\bibitem[CRT06]{CandesRT2006}
Emmanuel~J Candes, Justin~K Romberg, and Terence Tao, \emph{Stable signal
  recovery from incomplete and inaccurate measurements}, Communications on pure
  and applied mathematics \textbf{59} (2006), no.~8, 1207--1223.

\bibitem[DH13]{DemanetH13}
L.~{Demanet} and P.~{Hand}, \emph{{Recovering the Sparsest Element in a
  Subspace}}, October 2013, Arxiv preprint 1310.1654.

\bibitem[Don06]{Donoho2006}
David~L Donoho, \emph{Compressed sensing}, Information Theory, IEEE
  Transactions on \textbf{52} (2006), no.~4, 1289--1306.

\bibitem[DPS02]{doherty2002distinguishing}
Andrew~C Doherty, Pablo~A Parrilo, and Federico~M Spedalieri,
  \emph{Distinguishing separable and entangled states}, Physical Review Letters
  \textbf{88} (2002), no.~18, 187904.

\bibitem[EA06]{EladA2006}
Michael Elad and Michal Aharon, \emph{Image denoising via sparse and redundant
  representations over learned dictionaries}, Image Processing, IEEE
  Transactions on \textbf{15} (2006), no.~12, 3736--3745.

\bibitem[EP07]{ArgyriouEP06}
Andreas Argyriou~Theodoros Evgeniou and Massimiliano Pontil, \emph{Multi-task
  feature learning}, Advances in Neural Information Processing Systems 19:
  Proceedings of the 2006 Conference, vol.~19, MIT Press, 2007, pp.~41--48.

\bibitem[FH14]{frenkel14}
P{\'e}ter~E Frenkel and P{\'e}ter Horv{\'a}th, \emph{Minkowski's inequality and
  sums of squares}, Central European Journal of Mathematics \textbf{12} (2014),
  no.~3, 510--516.

\bibitem[FJK96]{FriezeJerrumKannan96}
Alan Frieze, Mark Jerrum, and Ravi Kannan, \emph{Learning linear
  transformations}, 37th {A}nnual {S}ymposium on {F}oundations of {C}omputer
  {S}cience ({B}urlington, {VT}, 1996), IEEE Comput. Soc. Press, Los Alamitos,
  CA, 1996, pp.~359--368. \MR{1450634}

\bibitem[For01]{Forster01}
J{\"u}rgen Forster, \emph{A linear lower bound on the unbounded error
  probabilistic communication complexity}, IEEE Conference on Computational
  Complexity, IEEE Computer Society, 2001, pp.~100--106.

\bibitem[Gri01]{Grigoriev01}
Dima Grigoriev, \emph{Linear lower bound on degrees of positivstellensatz
  calculus proofs for the parity}, Theor. Comput. Sci. \textbf{259} (2001),
  no.~1-2, 613--622.

\bibitem[GVX14]{GoyalVX13}
Navin Goyal, Santosh Vempala, and Ying Xiao, \emph{Fourier pca}, STOC, 2014,
  Also available as arXiv report 1306.5825.

\bibitem[Har70]{harshman1970foundations}
Richard~A Harshman, \emph{Foundations of the parafac procedure: Models and
  conditions for an" explanatory" multimodal factor analysis}.

\bibitem[Har07]{Harrison2007}
John Harrison, \emph{Verifying nonlinear real formulas via sums of squares},
  Theorem Proving in Higher Order Logics, Springer, 2007, pp.~102--118.

\bibitem[HG05]{henrion2005positive}
Didier Henrion and Andrea Garulli, \emph{Positive polynomials in control}, vol.
  312, Springer, 2005.

\bibitem[Hur91]{Hurwitz1891}
A.~Hurwitz, \emph{Ueber den vergleich des arithmetischen und des geometrischen
  mittels.}, Journal f\"ur die reine und angewandte Mathematik \textbf{108}
  (1891), 266--268, Available online at \url{http://eudml.org/doc/148823}.

\bibitem[Kru77]{kruskal1977three}
Joseph~B Kruskal, \emph{Three-way arrays: rank and uniqueness of trilinear
  decompositions, with application to arithmetic complexity and statistics},
  Linear algebra and its applications \textbf{18} (1977), no.~2, 95--138.

\bibitem[Las01]{Lasserre01}
Jean~B. Lasserre, \emph{Global optimization with polynomials and the problem of
  moments}, SIAM Journal on Optimization \textbf{11} (2001), no.~3, 796--817.

\bibitem[LCC07]{DeLathauwer07}
Lieven~De Lathauwer, Jos{\'e}phine Castaing, and Jean-Fran\c{c}ois Cardoso,
  \emph{Fourth-order cumulant-based blind identification of underdetermined
  mixtures}, IEEE Transactions on Signal Processing \textbf{55} (2007),
  no.~6-2, 2965--2973.

\bibitem[MLB{\etalchar{+}}08]{MairalLBHP2008}
Julien Mairal, Marius Leordeanu, Francis Bach, Martial Hebert, and Jean Ponce,
  \emph{Discriminative sparse image models for class-specific edge detection
  and image interpretation}, Computer Vision--ECCV 2008, Springer, 2008,
  pp.~43--56.

\bibitem[MRBL07]{RanzatoBL2007}
Y~Marc'Aurelio~Ranzato, Lan Boureau, and Yann LeCun, \emph{Sparse feature
  learning for deep belief networks}, Advances in neural information processing
  systems \textbf{20} (2007), 1185--1192.

\bibitem[Nes00]{Nesterov00}
Y.~Nesterov, \emph{Squared functional systems and optimization problems}, High
  performance optimization \textbf{13} (2000), 405--440.

\bibitem[NR09]{NguyenR09}
Phong~Q. Nguyen and Oded Regev, \emph{Learning a parallelepiped: Cryptanalysis
  of ggh and ntru signatures}, J. Cryptology \textbf{22} (2009), no.~2,
  139--160, Preliminary version in EUROCRYPT 2006.

\bibitem[OF96a]{OlshausenF96}
Bruno~A Olshausen and David~J Field, \emph{Emergence of simple-cell receptive
  field properties by learning a sparse code for natural images}, Nature
  \textbf{381} (1996), no.~6583, 607--609.

\bibitem[OF96b]{OlshausenF96b}
\bysame, \emph{Natural image statistics and efficient coding*}, Network:
  computation in neural systems \textbf{7} (1996), no.~2, 333--339.

\bibitem[OF97]{OlshausenF97}
Bruno~A. Olshausen and David~J. Field, \emph{Sparse coding with an overcomplete
  basis set: A strategy employed by v1?}, Vision Research \textbf{37} (1997),
  no.~23, 3311 -- 3325.

\bibitem[Par00]{Parrilo00}
Pablo~A Parrilo, \emph{Structured semidefinite programs and semialgebraic
  geometry methods in robustness and optimization}, Ph.D. thesis, California
  Institute of Technology, 2000.

\bibitem[Par06]{Parrilo2006polynomial}
\bysame, \emph{Polynomial games and sum of squares optimization}, Decision and
  Control, 2006 45th IEEE Conference on, IEEE, 2006, pp.~2855--2860.

\bibitem[Rez87]{Reznick87}
Bruce Reznick, \emph{A quantitative version of {H}urwitz{'} theorem on the
  arithmetic-geometric inequality}, J. reine angew. Math \textbf{377} (1987),
  no.~108-112.

\bibitem[Rez89]{Reznick89}
\bysame, \emph{Forms derived from the arithmetic-geometric inequality},
  Mathematische Annalen \textbf{283} (1989), no.~3, 431--464.

\bibitem[Sch08]{Schoenebeck08}
Grant Schoenebeck, \emph{Linear level {Lasserre} lower bounds for certain
  k-{CSPs}}, FOCS, 2008, pp.~593--602.

\bibitem[Sho87]{Shor87}
NZ~Shor, \emph{An approach to obtaining global extremums in polynomial
  mathematical programming problems}, Cybernetics and Systems Analysis
  \textbf{23} (1987), no.~5, 695--700.

\bibitem[SWW12]{SpielmanWW12}
Daniel~A. Spielman, Huan Wang, and John Wright, \emph{Exact recovery of
  sparsely-used dictionaries}, Journal of Machine Learning Research -
  Proceedings Track \textbf{23} (2012), 37.1--37.18.

\bibitem[Tuc66]{tucker1966some}
Ledyard~R Tucker, \emph{Some mathematical notes on three-mode factor analysis},
  Psychometrika \textbf{31} (1966), no.~3, 279--311.

\bibitem[YWHM08]{YangWHY2008}
Jianchao Yang, John Wright, Thomas Huang, and Yi~Ma, \emph{Image
  super-resolution as sparse representation of raw image patches}, Computer
  Vision and Pattern Recognition, 2008. CVPR 2008. IEEE Conference on, IEEE,
  2008, pp.~1--8.

\end{thebibliography}
